\documentclass[a4paper, 11pt]{article}
\usepackage[utf8]{inputenc}
\usepackage{amsmath}
\usepackage{amssymb}
\usepackage{amsthm}
\usepackage[backend=biber,style=alphabetic, maxnames = 10]{biblatex}
\usepackage{url}
\usepackage{hyperref}
\usepackage{xcolor}
\usepackage{dsfont}
\usepackage{float}
\usepackage{mathrsfs}
\usepackage{yfonts}

\makeatletter
\let\c@theorem\undefined
\let\c@lemma\undefined
\let\c@definition\undefined
\makeatother

\newtheorem{definition}{Definition}

\newtheorem{remark}{Remark}
\newtheorem{theorem}{Theorem}
\newtheorem{lemma}{Lemma}
\newtheorem{note}{Note}
\newtheorem{corollary}{Corollary}

\usepackage{graphicx}
 \usepackage{caption}
 \usepackage{subcaption}
 \usepackage{verbatim}
\usepackage[title]{appendix}
\usepackage{upgreek}
\usepackage{tikz}
\usetikzlibrary{shapes.geometric,plotmarks,backgrounds,fit,calc,circuits.ee.IEC}
\usetikzlibrary {arrows.meta}
\usetikzlibrary{patterns}
\usetikzlibrary{decorations.pathreplacing}
\usepackage{quantikz}

\usepackage{geometry}

\bibliography{references}

\newcommand{\bL}{\mathbf{L}}

\newcommand{\bD}{\mathbf{D}}

\newcommand{\cB}{{\cal B}}
\newcommand{\cC}{{\cal C}}

\newcommand{\cI}{{\cal I}}

\newcommand{\cM}{{\cal M}}
\newcommand{\cN}{{\cal N}}

\newcommand{\cP}{{\cal P}}

\newcommand{\cR}{{\cal R}}

\newcommand{\cT}{{\cal T}}
\newcommand{\cU}{{\cal U}}
\newcommand{\cV}{{\cal V}}
\newcommand{\cW}{{\cal W}}

\newcommand{\be}{\mathbf{e}}

\newcommand{\bu}{\mathbf{u}}

\newcommand{\CNOT}{\mathrm{CNOT}}

\newcommand{\ident}{\mathds{1}}
\newcommand{\dnorm}[1]{{\ensuremath{\|#1\|_\diamond}}}

\usepackage{braket}

\DeclareMathOperator{\tr}{Tr}
\DeclareMathOperator{\trans}{Trans}

\DeclareMathOperator{\cnot}{CNOT}

\title{Fault-tolerant quantum computation \\ with constant overhead for general noise}
\author{Matthias Christandl$^1$, Omar Fawzi$^2$, and Ashutosh Goswami$^1$ \\[2mm]
    {\small $^1$Department of Mathematical Sciences, University of Copenhagen, Denmark} \\
    {\small $^2$Universit\'e de Lyon, Inria, ENS de Lyon, UCBL, LIP, France}\\[0.5em]
  {\small christandl@math.ku.dk, omar.fawzi@ens-lyon.fr, akg@math.ku.dk}}
\date{}

\begin{document}

\title{Fault-tolerant quantum computation \\ with constant overhead for general noise}

\maketitle
\begin{abstract}
Fault-tolerant quantum computation traditionally incurs substantial resource overhead, with both qubit and time overheads scaling polylogarithmically with the size of the computation. While prior work by Gottesman showed that constant qubit overhead is achievable under stochastic noise using quantum low-density parity-check (QLDPC) codes, it has remained an open question whether similar guarantees hold under more general, non-stochastic noise models. In this work, we address this question by considering a general circuit-level noise model defined via the diamond norm, which captures both stochastic and non-stochastic noise, including coherent and amplitude damping noise. We prove that constant qubit overhead fault-tolerant quantum computation is achievable in this general setting, using QLDPC codes with constant rate and linear minimum distance. To establish our result, we develop a fault-tolerant error correction scheme and a method for implementing logic gates under general circuit noise. These results extend the theoretical foundations of fault-tolerant quantum computation and
offer new directions for fault-tolerant architectures under realistic noise models.
\end{abstract}

\maketitle

\section{Introduction}
Fault tolerance is essential for large-scale quantum computations due to the inherently noisy nature of quantum systems~\cite{preskill1998fault, gottesman2010introduction}. A leading approach to fault-tolerant quantum computation is based on active error correction, where logical qubits encoded in the code space of an error-correcting code are created using many noisy physical qubits. Using fault-tolerant procedures, one can realize logic gates in the code space of the error-correcting code. The celebrated threshold theorem states that by doing active error correction, it is possible to realize arbitrarily long finite quantum computations with classical input and output, given that the noise rate in the quantum circuit is below a threshold value~\cite{aharonov1997fault, kitaev, Aliferis2006quantum}.   
Fault-tolerance, however, comes at a price of increased qubit count compared to the noiseless protocol. More precisely, to realize a quantum circuit, one replaces it with a fault-tolerant quantum circuit, which operates on a much larger number of qubits. Therefore, it is of interest from a practical standpoint to minimize the qubit overhead, that is, the ratio of the number of qubits in fault-tolerant circuit and original circuit.

\smallskip Fault-tolerant quantum computing based on concatenated codes, which were the first error-correcting codes used to prove the threshold theorem, have a polylogarithmic asymptotic qubit overhead~\cite{aharonov1997fault, kitaev, Aliferis2006quantum}, meaning that the ratio of physical to logical qubits grows polylogarithmically with the size of the quantum computation. This naturally raises the question whether fault-tolerant quantum computation is possible with a qubit overhead that remains constant with the problem size. For stochastic noise, this question was partially settled in~\cite{gottesman2013fault}, where a construction of constant overhead fault-tolerant quantum computation was provided using quantum low-density parity-check (QLDPC) codes.
The construction in~\cite{gottesman2013fault} was based on the assumption that a family of QLDPC codes with certain properties exists. More precisely, it supposes that there exists a family of constant-rate QLPDC codes that can suppress random errors with noisy syndrome extraction and that furthermore have an efficient classical decoding algorithm. It was then shown in~\cite{fawzi2020constant} that quantum expander codes satisfy all the requirements needed in Gottesman's construction; hence, establishing constant overhead fault tolerance. It is important to note that Gottesman's construction leads to an increased overhead in depth due to a sequentialization step requiring that only a constant number of logic gates are applied at each time step. In particular, the overhead in depth is polynomial. In~\cite{yamasaki2024time}, a different approach based on concatenated codes is used to show constant overhead fault-tolerant quantum computation, where the time overhead is quasi polylogarithmic (see also Note~\ref{note}).  These results~\cite{gottesman2013fault, fawzi2020constant, yamasaki2024time} have generated significant interest in constant-overhead quantum computation, particularly those based on QLDPC codes. Subsequent works have focused on optimizing these constant-overhead constructions for practical implementation~\cite{xu2024constant, bravyi2024high, cowtan2025parallel, bonilla2025constant, yoshida2025concatenate}.

\smallskip 
However, all the previous works on constant overhead fault-tolerant quantum computation~\cite{gottesman2013fault, fawzi2020constant, yamasaki2024time, xu2024constant,bravyi2024high, cowtan2025parallel, bonilla2025constant, yoshida2025concatenate} consider \emph{stochastic} circuit noise (e.g., Pauli and local stochastic noise), 
leaving the question open whether it is possible to achieve constant qubit overhead under more realistic non-stochastic circuit noise (e.g., coherent and amplitude damping noise). Moreover, since the threshold theorem based on concatenated codes also applies to non-stochastic circuit noise~\cite{kitaev, aharonov1997fault}, it is natural to ask whether such generalizations also hold for constant overhead constructions. 

\smallskip In this work, we consider a general circuit noise, which includes both coherent noise and stochastic noise, as special cases. In particular, general noise with parameter $\delta$ replaces each gate $g$ in the circuit by an arbitrary quantum channel $\tilde{g}$, which is $\delta$-close to $g$ in diamond norm, that is, $\dnorm{g - \tilde{g}} \leq \delta$~\cite{kitaev}. We show that constant overhead can be achieved for general circuit noise using QLDPC codes that have constant rate and linear minimum distance. It's worth noticing that linear minimum distance is not a requirement for stochastic circuit noise~\cite{gottesman2013fault, fawzi2020constant}. However, recent breakthroughs have demonstrated the existence of such QLDPC codes~\cite{panteleev2022asymptotically, leverrier2022quantum}; therefore, our results do not rely on any unproven assumptions. 

\smallskip To establish constant overhead fault-tolerant quantum computation for general circuit noise, we revisit the construction from~\cite{gottesman2013fault, fawzi2020constant}, and develop necessary tools and generalizations to establish constant overhead fault-tolerant quantum computation for general circuit noise. In particular, the following two points are necessary for constructions in~\cite{gottesman2013fault, fawzi2020constant}:
\begin{itemize}
    \item[1] Fault-tolerant error correction for QLPDC codes can be done by measuring stabilizer generators of the code and performing error correction based on the measurement outcomes (syndrome).

    \item[2] The logic gates can be implemented using gate teleportation, where the required ancilla state in the code space can be prepared using concatenated codes.
\end{itemize}
\begin{figure}
\centering
\begin{quantikz}[row sep=0.5cm, column sep=0.2cm]
 & \ctrl{2} & \qw & \qw & \qw & \ctrl{1}  & \qw  & \qw & \qw & \qw & \qw & \qw & \qw \\
 & \qw  & \ctrl{2} & \qw & \qw & \targ{} & \qw & \qw & \qw & \qw & \qw & \qw & \qw\\
 & \targ{} & \qw & \qw & \qw &  \ctrl{2}  & \qw  & \qw & \qw & \qw & \ctrl{1} & \qw & \qw  \rstick[5]{future light cone of the \\ highlighted CNOT gate}  \\
 & \qw & \targ{} & \qw & \qw & \qw & \qw & \qw & \qw & \qw & \targ{} & \qw & \qw \\
& \ctrl{1} \gategroup[2,steps=1,style={dashed,rounded
corners,fill=blue!20, inner
xsep=2pt},background,label style={label
position=above,anchor=north,yshift=0.2cm}]{} & \qw  & \qw & \qw & \targ{} & \qw & \qw & \qw & \qw & \ctrl{1} & \qw & \qw\\
 & \targ{} & \qw  &  \qw & \qw & \qw & \qw & \qw & \qw & \qw & \targ{} & \qw & \qw 
\end{quantikz}
\caption{The figure illustrated the future light cone of the highlighted CNOT gate, i.e., the set of qubits that are affected by the $\cnot$ gate.}
\label{fig:light-cone}
\end{figure}
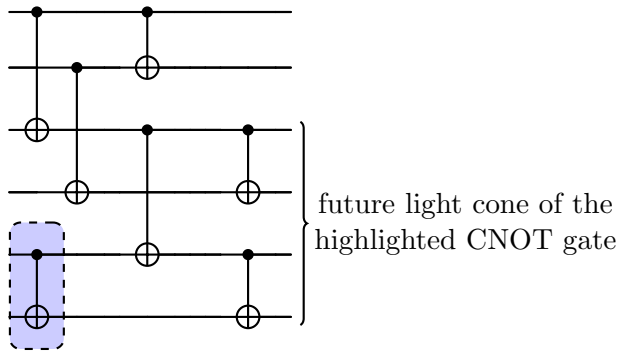
We generalize the above two points for general circuit noise, considering QLDPC codes with constant overhead and linear minimum distance.

\smallskip For the first point, we consider parallel error correction of several code blocks of QLDPC codes and 
analyze the syndrome extraction circuit for each code block under general circuit noise. 

\smallskip We divide the syndrome extraction $\Phi_{\mathrm{syn}}$ circuit in three parts, (1) initialization of ancilla qubits in a Pauli-$X$ or $Z$ basis state, (2) application of CNOT or CZ gates to entangle data qubits with ancilla qubits, (3) Measurement of the ancilla qubits in Pauli-$X$ or $Z$ basis (see also Fig.~\ref{fig:syn-ext-circuit}). We show that the errors that happen during $\Phi_{\mathrm{syn}}$ can be pushed just before the last layer of the circuit, where measurements are applied. Note that errors multiply as they are pushed since they can propagate through two qubit gates. However, error propagation is limited since $\Phi_{\mathrm{syn}}$ has constant depth (with respect to the code length) for QLDPC codes, implying that the future light cone of any gate only contains a constant number of qubits (see Fig.~\ref{fig:light-cone} for an illustration of the future light cone and Lemma~\ref{lem:pushing-errs}, where a detailed analysis of $\Phi_{\mathrm{syn}}$ under general circuit noise is given). 

\smallskip We then combine our analysis of the syndrome extraction circuit with the single-shot classical decoder introduced in~\cite{gu2024single} to show that the parallel error correction of several code blocks controls the errors in each code block, so that the computation can be sustained for arbitrarily long time (see Theorem~\ref{thm:cir-noise-ftec} for more details). It is worth emphasizing that parallel error correction is necessary since in our construction logical qubits are encoded in many smaller code blocks of QLDPC codes and not just in one big code block, similarly to~\cite{gottesman2013fault, fawzi2020constant}(see also Fig.~\ref{fig:layout-sim-circ}). We note that, in this regard, our analysis improves previous works~\cite{gottesman2013fault, fawzi2020constant}, where error correction of only one code block was considered, thereby ignoring correlations that arise between different blocks. 

\smallskip For the second point, we construct fault-tolerant logic gates by combining a fault-tolerant state preparation method for general circuit noise given in~\cite{christandl2024fault}, together with gate teleportation. By integrating our generalizations of fault-tolerant error correction and logic gates within the constructions of~\cite{gottesman2013fault,fawzi2020constant}, we establish the following theorem.
\begin{theorem} [Informal, see Theorem~\ref{thm:main-cons-over} for details] \label{thm:main-cons-over-informal}
Consider a quantum circuit $\Phi$, with classical input and output, working on $x$ qubits, and of size (i.e., total number of locations) $\mathrm{poly}(x)$. For any asymptotic overhead $\alpha > 1$, there exists a constant error rate $\delta_{th} > 0$, and a quantum circuit $\overline{\Phi}$, with the same input and output systems as $\Phi$, and working on $x' = \alpha x + o(x)$ qubits, such that any noisy version $[\overline{\Phi}]_{\delta}$, under general circuit noise with parameter $\delta \leq \delta_{th}$, simulates $\Phi$ with a vanishing error probability as $x$ goes to infinity.  
\end{theorem}

 \begin{figure}[!t]
    \centering
\begin{tikzpicture}

\tikzset{
  half circle/.style={
      semicircle,
      shape border rotate=90,
      anchor=chord center
      }
}

\draw 
(0.0, 0) node[inner xsep=0.45cm, inner ysep=0.5cm, draw, fill=blue!20] (a1) {EC}
(0.0, -2.5) node[inner xsep=0.45cm, inner ysep=0.5cm, draw, fill=blue!20] (b1) {EC}
;

\draw
($(a1.west) + (-0.6, 0.4)$) node[](w1){} 
($(a1.west) + (-0.6, -0.4)$) node[](w2){}  
(w1) to (w1  -| a1.west)
(w2) to (w2  -| a1.west)
;

\draw
($(b1.west) + (-0.6, 0.4)$) node[](w3){} 
($(b1.west) + (-0.6, -0.4)$) node[](w4){}  
(w3) to (w3  -| b1.west)
(w4) to (w4  -| b1.west)
;

\draw 
(2.0, 0) node[inner xsep=0.45cm, inner ysep=0.5cm, draw, fill=blue!20] (a2) {EC}
(2.0, -2.5) node[inner xsep=0.45cm, inner ysep=0.5cm, draw, fill=blue!20] (b2) {EC}
;

\draw 
(6, 0) node[inner xsep=0.4cm, inner ysep=0.5cm, draw, fill=blue!20] (a3) {EC}
(6, -2.5) node[inner xsep=0.4cm, inner ysep=0.5cm, draw, fill=blue!20] (b3) {EC}
;

\draw
(w1 -| a1.east) to (w1 -| a2.west)
(w2 -| a1.east) to (w2 -| a2.west)
(w1 -| a2.east) to ++(0.5, 0)++ (0.9,0) node[](){$\cdots$} 
(w2 -| a2.east) to ++(0.5, 0) ++ (0.9,0) node[](){$\cdots$} 
(w1 -| a3.west) to ++(-0.5, 0) 
(w2 -| a3.west) to ++(-0.5, 0)
(w3 -| b1.east) to (w3 -| b2.west)
(w4 -| b1.east) to (w4 -| b2.west)
(w3 -| b2.east) to  ++(0.5, 0) ++ (0.9,0) node[](){$\cdots$} 
(w4 -| b2.east) to  ++(0.5, 0) ++ (0.9,0) node[](){$\cdots$} 
(w3 -| b3.west) to ++(-0.5, 0) 
(w4 -| b3.west) to ++(-0.5, 0) 
;

\draw[dotted, thick]
($(a1.west) + (-0.25, 0)$) node[rotate = 90](){$\cdots$} 
($(b1.west) + (-0.25, 0)$) node[rotate = 90](){$\cdots$} 
;

\draw 
(8.5, 0) node[inner xsep=0.45cm, inner ysep=0.5cm, draw, fill=blue!20] (a4) {EC}
(8.5, -3.25) node[inner xsep=0.45cm, inner ysep=1.5cm, draw, fill=blue!20] (b4) {GT}
;

\draw
(-0.5, 0.4) node[] (w1) {}
(-0.5, -0.4) node[] (w2) {}
(w1 -| a3.east) to (w1 -| a4.west)
(w1 -| a4.east) to ++ (0.5, 0)
(w2 -| a3.east) to (w2 -| a4.west)
(w2 -| a4.east) to ++ (0.5, 0)
;

\draw
(-0.5, -2.1) node[] (w3) {}
(-0.5, -2.9) node[] (w4) {}
(w3 -| b3.east) to (w3 -| b4.west)
(w3 -| b4.east) to ++ (0.5, 0)
(w4 -| b3.east) to (w4 -| b4.west)
(w4 -| b4.east) to ++ (0.5, 0)
;

\draw
(1.4, -6.0) node[inner xsep=0.3cm, inner ysep=2.0cm, draw, fill=red!20] (p1) {}
(6, -6.0) node[inner xsep=0.3cm, inner ysep=1.0cm, draw, fill=red!20] (p2) {}
($0.5*(p1) + 0.5*(p2)$) node[inner xsep=0.3cm, inner ysep=1.5cm, draw, fill=red!20] (p3) {}
(p1) node[rotate = 90](){}
(p2) node[rotate = 90](){}
(p3) node[rotate = 90](){}

($(p1.west) + (-1.0, 1.9)$) node[draw, half circle, fill=red!20] (i1){}
($(p1.west) + (-1.0, -1.9)$) node[draw, half circle, fill=red!20] (i2){}
($0.5*(i1) + 0.5*(i2) + (0.07, 0)$) node[rotate = 90](){}
($(i1) + (0, -1.0)$) node[rotate = 90](){$\dots$}
($(i2) + (0, 1.0)$) node[rotate = 90](){$\dots$}
($0.6*(p1) + 0.4*(p3)$) node[](){$\cdots$}
($0.6*(p3) + 0.4*(p2)$) node[](){$\cdots$}
($(p3.west) + (-.5, 1.3)$) to ($(p3.west) + (0, 1.3)$)
($(p3.west) + (-.5, -1.3)$) to ($(p3.west) + (0, -1.3)$)
($(p2.west) + (-.5, -0.9)$) to ($(p2.west) + (0, -0.9)$)
($(p2.west) + (-.5, 0.9)$) to ($(p2.west) + (0, 0.9)$)
($(p2.east) + (0, 0.9)$) to ($(p2.east) + (0.3, 0.9)$) to ++(0, 1.6) node[inner sep = 0, outer sep = 0](y1){}
($(p2.east) + (0, -0.9)$) to ($(p2.east) + (0.6, -0.9)$) to ++(0, 2.2) node[inner sep = 0, outer sep = 0](y2){}
(y1) to (y1 -| b4.west)
(y2) to (y2 -| b4.west)
($(p3.west) + (-0.25, 0.8)$) node[rotate = 90](){$\cdots$}
($(p3.west) + (-0.25, 0.0)$) node[rotate = 90](){$\cdots$}
($(p3.west) + (-0.25, -0.8)$) node[rotate = 90](){$\cdots$}
($(p2.west) + (-0.25, 0.35)$) node[rotate = 90](){$\cdots$}
($(p2.west) + (-0.25, -0.35)$) node[rotate = 90](){$\cdots$}
;

\draw
(i1) to (i1 -| p1.west)
(i2) to (i2 -| p1.west)
;

\draw

;

\draw[dashed, thick] 
(7.3, 1.5) to (7.3, -8.4) 
;

\draw
(3, 1.2) node[](){{\bf Ancilla preparation}}
(9.3, 1.2) node[](){{\bf Gate teleportation}}
;

\draw[decorate, decoration = brace]
($(i2) + (-1.2, -0.2)$) to ($(i1) + (-1.2, 0.2)$)  
;

\draw[decorate, decoration = brace]
(-1.3, -3.6) to ++(0, 4.2)
;

\draw
($(i1) + (-1.4, 0.2)$) node[left, rotate = 90](){Ancilla preparation circuit}
;

\draw
(-1.55, 0.2) node[left, rotate = 90](){Data code blocks}
;

\end{tikzpicture}
\caption{This figure illustrates the realization of a logic gate in the fault-tolerant circuit $\overline{\Phi}$ corresponding to $\Phi$ according to Theorem~\ref{thm:main-cons-over-informal}.  The blocks corresponding to EC (error correction) and GT (gate teleportation) represent code blocks of the QLDPC codes in $\overline{\Phi}$. The size of each block is $m = O(\sqrt{x})$, where $x$ is the number of qubits in $\Phi$. The logic gate is realized in two steps, namely ancilla preparation and gate teleportation steps. During the ancilla preparation, a code state is prepared, using state preparation method from Theorem~\ref{thm:state-prep}. Note that the \emph{decreasing size} of the layers in the ancilla preparation circuit reflect the fact that extra qubits are needed during ancilla preparation. The qubit overhead in ancilla preparation circuit alone is in fact not constant (it is polylogarithmic). While the ancilla state is prepared, repeated rounds of error correction are applied on the data code blocks. During the gate teleportation, the prepared ancilla state is transferred to the specified code block, and the logic gate is realized using gate teleportation. On the remaining code blocks error correction is applied.
}
\label{fig:layout-sim-circ}
\end{figure}
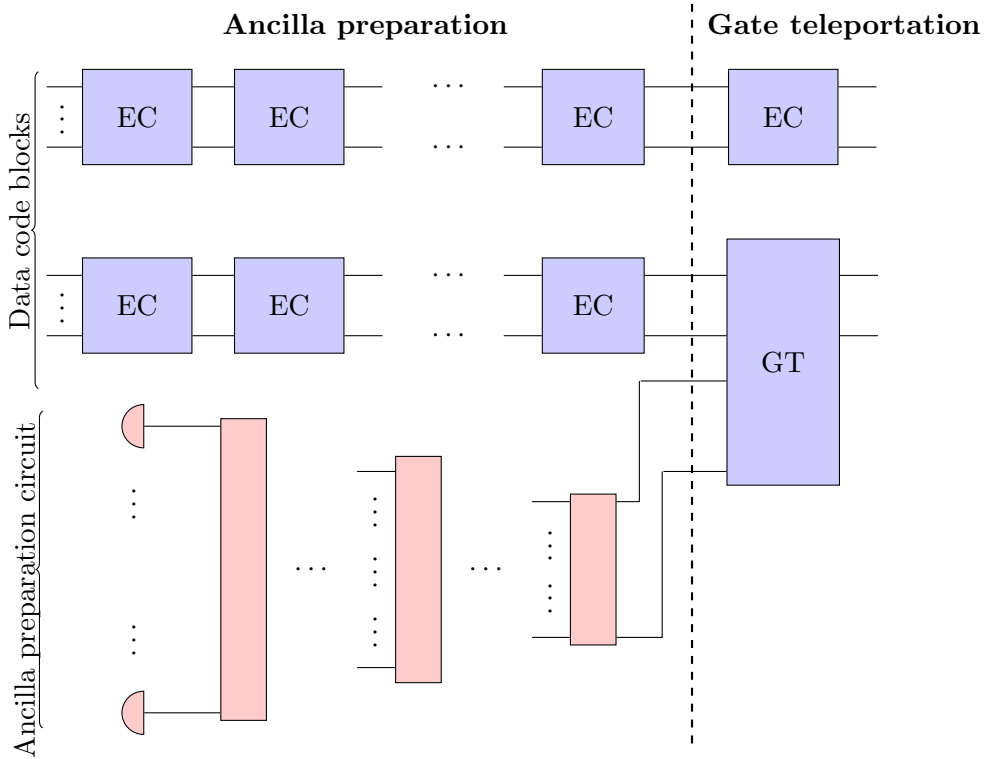

To construct the fault-tolerant circuit $\overline{\Phi}$ according to Theorem~\ref{thm:main-cons-over-informal}, we consider Gottesman's fault-tolerant scheme with constant overhead~\cite{gottesman2013fault, fawzi2020constant}. Gottesman's construction involves two different types of codes:

\begin{itemize}
    \item[$(i)$] QLDPC codes, which store the logical qubits on which quantum computation is performed.
    \item[$(ii)$] Concatenated codes, which are used only for preparing certain entangled ancillary states in the QLDPC code space. These states are used for fault-tolerant logic gate implementation on QLDPC code via gate teleportation.
\end{itemize}
We note that the overhead in ancillary state preparation using concatenated codes is not constant, i.e., it grows with code size. To make sure, that qubit overhead remains constant with respect to $x$, i.e., the number of qubits in $\Phi$, a sequentialization technique is used where $x$ logical qubits are encoded into multiple blocks of a QLDPC code and logical gate is applied on only one (or a small number) of these blocks (see also Fig.~\ref{fig:layout-sim-circ}).

\smallskip In our generalization of Gottesman's protocol for general noise models, we assume the availability of a single-shot error correction for the QLDPC code, with constant-time complexity. However, single-shot error correction is not supposed for the concatenated codes. Fault-tolerant preparation of ancillary states using concatenated codes has been analyzed in~\cite{christandl2024fault}. The state preparation procedure is relatively slow, in the sense that its time complexity grows with the size of the QLDPC code. Nevertheless, this does not conflict with the use of single-shot error correction on QLDPC code. This is because while the ancillary state is being prepared, repeated rounds of error correction are applied on the QLDPC code (see also Fig.~2). These repeated rounds of error correction ensure that logical information remains protected in QLDPC code throughout the preparation process.

Below, we briefly describe the construction of the fault-tolerant quantum circuit $\overline{\Phi}$ according to Theorem~\ref{thm:main-cons-over-informal}.

We divide $x$ qubits of circuit $\Phi$ into $x/m$ blocks, each containing $m = O(\sqrt{x})$ qubits.  We suppose a sequential realization of $\Phi$, where at each layer (time step), only one non-trivial gate is applied on a selected set of qubits, and the idle gate is applied on the remaining qubits. To construct $\overline{\Phi}$,  we consider a QLDPC code that encodes $m$ logical qubits in $n$ physical qubits. We replace each block of $\Phi$, containing $m$ qubits, by $n$-qubits corresponding to the QLDPC code. Then $\overline{\Phi}$ is obtained by replacing  non-trivial gates in the sequential circuit $\Phi$ by the corresponding logic gate based on gate teleportation. The idle gates are replaced by multiple rounds of single-shot error correction, taking into account the time-complexity of ancillary state preparation within gate teleportation. The realization of a logic gate through gate teleportation is illustrated in Fig.~\ref{fig:layout-sim-circ}. Note that logic gates in $\overline{\Phi}$ need to be applied sequentially due to the fact that the qubit overhead in ancilla preparation is not constant but poly logarithmic (see also Fig.~\ref{fig:layout-sim-circ}). This means that for one non-trivial logic gate, we need $O(\sqrt{x} \: \mathrm{polylog} \:\frac{x}{\epsilon})$ ancilla qubits, where $\epsilon$ is the accuracy with which we wish to realize $\Phi$. The sequential implementation ensures that ancilla qubits can be reused or reinitialized, hence keeping the total qubit overhead constant. Applying parallel logic gates would mean preparing ancilla states in parallel, hence blowing up the overhead with respect to $x$.

The remaining sections have been organized as follows. In Section~\ref{sec:prelim}, we give basic definitions and review the state preparation theorem from~\cite{christandl2024fault}. In Section~\ref{sec:proof}, we provide a proof of Theorem~\ref{thm:main-cons-over-informal}, assuming a definition of fault-tolerant error correction. Later in Section~\ref{sec:single-shot-err}, we use the single shot decoding from~\cite{gu2024single} to show that the definition of fault-tolerant error correction is satisfied for QLDPC codes with constant rate and linear minimum distant, assuming general circuit noise during syndrome extraction.
\begin{note} \label{note}
The present manuscript constitutes part of the arXiv submission~\cite{christandl2024fault}. We note that shortly after our work was posted, two related works~\cite{nguyen2024quantum, tamiya2024polylog} appeared on arXiv, showing that for local stochastic noise, the time overhead can be further reduced to polylogarithmic, while still maintaining constant qubit overhead.
\end{note}
\section{Preliminary} \label{sec:prelim}
In this section, we provide some basic definitions, outline the circuit model of quantum computation and describe general circuit noise, used to derive the results in this paper. Moreover, we review the state preparation theorem from~\cite{christandl2024fault}, which will be used later in the proof of constant overhead fault-tolerant quantum computation.

\paragraph{Notation.} We use symbols $M, L, N$ for finite dimensional Hilbert spaces, and $\bL(M, N)$ for linear operators from $M$ to $N$ or simply $\bL(M)$ for linear operators on $M$. We use $\bD(M)$ for quantum states on $M$. A qubit Hilbert space is denoted by $\mathbb{C}^2$. Quantum channels are denoted by the calligraphic letters $\cT, \cP, \cR, \cN, \cV, \cW$, etc. The identity operator is denoted by $\ident$ and the identity quantum channel by $\cI$.

\subsection{Basic definitions}

\medskip Let $[n] := \{1, \dots, n\}$ be a set of labels for $n$-qubits. For any $A \subseteq [n]$, we define the space of errors with support on $A$ as the space of linear operators acting only on qubits in $A$, i.e.,

\begin{equation} \label{eq:map-A}
  \mathscr{E}(A) := \{ (\otimes_{i \in [n] \setminus A} \ident_i) \otimes E_A : E_A \in \bL(M^{\otimes |A|})\} \subseteq \bL(M^{\otimes n}),
\end{equation}
 Using $\mathscr{E}(A)$, we define the space of operators with \emph{weight} $z \leq n$ qubits as follows.
\begin{definition}[Weight of an operator] \label{def:weight-z}
The space of operators with weight $z \leq n$ is defined as
   \begin{equation} \label{eq:weight-z}
    \mathscr{E}(n, z) := \mathrm{span}\{ E : E \in \mathscr{E}(A) \text{ for some } A \subseteq [n] \text{ s.t. } |A| \leq z \}
\end{equation}
\end{definition}
\begin{definition}[Weight of a superoperator] \label{def:weight}
 The space of superoperators with weight $z \leq n$ is defined as, 
\begin{equation} \label{eq:weight-z-sup}
  \mathscr{E}(n,z) (\cdot) \mathscr{E}(n, z)^*  = \mathrm{span}\{ E_1 (\cdot) E_2 : E_1 \in \mathscr{E}(n,z), E_2 \in \mathscr{E}(n,z)^*\}.
\end{equation}   
\end{definition}
 We note that in Def.~\ref{def:weight}, the operators $E_1$ and $E_2$ both have weight $z$ but they may be supported on completely different subset of the $n$ qubits. We have adopted this broader definition so as not to restrict error maps to quantum channels, but rather consider them as general superoperators. This broader definition simplifies the analysis of fault-tolerant error correction, under general noise models. 

\smallskip We will use the diamond norm as a measure of distance between two superoperators~\cite{kitaev, watrous2009semidefinite, kretschmann2008information}.
\begin{definition}[Diamond norm]
 For any superoperator $\cT: \bL(M) \to \bL(N)$, the diamond norm $\dnorm{T}$ is defined as,
\begin{equation}
    \dnorm{\cT} := \sup_{G} \| \cI_G \otimes \cT  \|_1,
\end{equation}
where $G$ is a reference system and the one-norm is given by,
\begin{equation}
 \| \cT  \|_1   := \sup_{\substack{\| \rho\|_1 \leq 1 \\ \rho \in \bL(M)}}  \| \cT(\rho) \|_1. 
\end{equation}    
\end{definition}

\subsection{Quantum circuits and noise model}
We consider joint quantum computation on classical and quantum bits. Any quantum computation is performed using a quantum circuit, which is constructed from a gate set acting on classical and quantum bits. We consider a gate set $\mathbf{A}$, containing a universal set of gates on classical bits, and a set of quantum gates, including state preparation, quantum measurement, a universal set of  unitary gates, and the corresponding classically controlled unitary gates.

\smallskip As in~\cite[Section~2]{christandl2024fault}, we consider a hybrid model of a quantum gate, acting on classical and quantum bits. We divide classical bits as control bits and operational bits. Let $l$ be the number of control bits and $o, o'$ be the number of operational bits at the input and output, respectively. Let $q, q'$ be the number of qubits at the input and output, respectively. For $\bu \in \{0, 1\}^l$, define a classical linear map $\mathcal{B}^{\bu}: \mathbb{B}^{\otimes l} \to \mathbb{B}^{\otimes l}$, $\mathcal{B}^{\bu} (\cdot) :=  \proj{\bu}(\cdot) \proj{\bu}$. Then, the gate $g$ realizes the quantum channel
\begin{equation} \label{eq:realization-qgate}
\cT_g := \sum_{\mathbf{u} \in \{0, 1\}^l} \mathcal{B}^\mathbf{u}  \otimes \cT^\mathbf{u},
\end{equation}
where $\cT^{\bu}$ is a quantum channel from $ \mathbb{B}^{\otimes o} \otimes  \bL((\mathbb{C}^2)^{\otimes q})$ to $  \mathbb{B}^{\otimes o'} \otimes \bL((\mathbb{C}^2)^{\otimes q'})$. Here, $\cT_g$ signifies transformation associated with the gate $g$.

\smallskip We consider a definition of quantum circuit based on layered implementation of gates, as given in~\cite{kitaev}(see also~\cite[Section~2]{christandl2024fault}).

\begin{figure}[t!]
\centering
\begin{quantikz}[row sep=1.0em, column sep=1em, wire types={q,c}]
\qw & \gate{\cT_{g_{1,1}}} \gategroup[7,steps=1,style={dashed,rounded
corners,fill=blue!20, inner
xsep=2pt},background,label style={label
position=above,anchor=north,yshift=0.25cm}]{$i = 1$} & &  & &\gate{\cT_{g_{2,1}}}  \gategroup[7,steps=1,style={dashed,rounded
corners,fill=blue!20, inner
xsep=2pt},background,label style={label
position=above,anchor=north,yshift=0.25cm}]{$i= 2$} &  &  \gate[2]{\cT_{g_{3,1}}}  \gategroup[7,steps=1,style={dashed,rounded
corners,fill=blue!20, inner
xsep=2pt},background,label style={label
position=above,anchor=north,yshift=0.25cm}]{$i = 3$}  &  \ \ldots\ & \gate[2]{\cT_{g_{d, 1}}} \gategroup[7,steps=1,style={dashed,rounded
corners,fill=blue!20, inner
xsep=2pt},background,label style={label
position=above,anchor=north,yshift=0.25cm}]{$i = d$} & \qw \\
\setwiretype{c} \cw & \gate[3]{\cT_{g_{1,2}}} & & & & \gate{\cT_{g_{2,2}}} &  &  &\ \ldots\ &  & \qw \\
\setwiretype{c} &   &  &  & & \gate[2]{\cT_{g_{2,3}}} & \trash{\text{trace}} \\
\setwiretype{q} &   & & & &  & & \gate[2]{\cT_{g_{3,2}}} & \trash{\text{trace}}\\[1.5em]
\setwiretype{q} & \gate[]{\cT_{g_{1,3}}} & & & & \gate{\cT_{g_{2,4}}}  &  & &\ \ldots\ & \gate[1]{\cT_{g_{d,2}}} & \\[1.5em]
\wave&&&&&&&&&&\\[1.5em]
 & \gate{\cT_{g_{1, s_1}}} & & & &
\gate{\cT_{g_{2, s_2}}} & & \gate{\cT_{g_{3,s_3}}} & \ \ldots\  & \gate{\cT_{g_{d,s_d}}}   & \qw
\end{quantikz}
\caption{This figure illustrates an example of quantum circuit according to Def.~\ref{def:quantum_circuit}.}
\label{fig:qcircuit-ex}
\end{figure}
\begin{definition}[Quantum Circuit] \label{def:quantum_circuit}
    A quantum circuit $\Phi$ of depth $d$ is a collection of the following objects (see also Fig.~\ref{fig:qcircuit-ex}).
\begin{itemize}
    \item[$(1)$]   A sequence of finite sets $\Delta_0, \dots, \Delta_d$, called \emph{layers}. The zeroth layer $\Delta_0$ is called the \emph{input} of the circuit, and the final layer $\Delta_d$ is the \emph{output}. Each layer $\Delta_i$ (for $i = 0, \dots, d$) decomposes as $\Delta_i = \Delta_i^c \cup \Delta_i^q$, where $\Delta_i^c$ contains \emph{classical wires} and $\Delta_i^q$ contains \emph{quantum wires}. If the circuit has only classical input, then $\Delta_0^q = \emptyset$; if the circuit has no input, then $\Delta_0 = \emptyset$.
    
\item[$(2)$] Partitions of each set $\Delta_{i-1}$, $i = 1, \dots, d$ into ordered subsets (registers) $A_{i1}, \dots, A_{is_i}$, and a corresponding partition of a superset $\Delta'_i \supseteq \Delta_i$ into registers $A'_{i1}, \dots, A'_{is_i}$. Each register $A_{ij}, A'_{ij}$ contains $l_{ij}$ control bits, $q_{ij}, q'_{ij}$ qubits,  and $o_{ij}, o'_{ij}$ operational bits, respectively.

    \item[$(3)$] A collection of gates $g_{i,j}$ drawn from the gate set $\mathbf{A}$, each realizing a quantum channel 
    \begin{equation*}
    \cT_{g_{i,j}}: \mathbb{B}^{
    \otimes {l_{ij}}} \otimes \mathbb{B}^{\otimes {o_{ij}}} \otimes \bL((\mathbb{C}^2)^{\otimes q_{ij}}) \to \mathbb{B}^{\otimes {l_{ij}}} \otimes \mathbb{B}^{\otimes {o'_{ij}}} \otimes \bL((\mathbb{C}^2)^{\otimes q'_{ij}}).    
    \end{equation*}
 The gate $g_{i,j}$ or their numbers $(i,j)$ are called the locations of the quantum circuit. 
    Note that for a fixed $i, 1 \leq i \leq d$, the set of locations $\{g_{i,j}: 1 \leq j \leq s_j \}$ corresponds to the locations in a single layer of the quantum circuit.
\end{itemize}
The \emph{size} of $\Phi$, denoted by $|\Phi|$, is defined as the total number of \emph{locations} it contains.
\end{definition}
\noindent A quantum circuit $\Phi$ according to Definition~\ref{def:quantum_circuit} realizes a quantum channel 
\begin{equation}
    \cT = \cT_d \circ \cdots \circ \cT_1,
\end{equation}
 where $\cT_i$ acts as follows on any  $\rho \in \bL({\mathbb{C}^2}^{\otimes |\Delta_{i-1}|})$,
\begin{equation}
    \cT_i (\rho) =  \tr_{\Delta'_i \setminus \Delta_i}\Big[\cT_{g_{i,1}}[A'_{i1}; A_{i1}] \otimes \dots \otimes \cT_{g_{i, s_i}}[A'_{is_i}; A_{is_i}] \Big] \left(\rho \right),
\end{equation}
where $\cT_{g_{i, j}}[A'_{ij}; A_{ij}]$ denotes the quantum channel $\cT_{g_{i,j}}$, applied to the register $A_{ij}$, whose output is stored in the register $A'_{ij}$.
\begin{remark} \label{rem:idle}
\smallskip Note that in Def.~\ref{def:quantum_circuit}, all the qubits in $\Delta_{i-1}, i =1, \dots, d$, are acted upon by a gate at the $i^{th}$ layer. We suppose that the idle gate acts on each waiting qubit. This is natural for capturing the effect of noise on the waiting qubits (see Section~\ref{sec:noise-qcirc} below).    
\end{remark}

\subsection{General circuit noise} \label{sec:noise-qcirc}
Let $\delta \in [0, 1]$ be an error rate. When general circuit noise with an error rate $\delta$ acts on a quantum circuit $\Phi$, we have that
\begin{enumerate}
    \item Purely classical elements of $\Phi$, i.e. the gates with only classical input-output, are implemented perfectly without any error.

    \item  A gate $g \in \mathbf{A}$ with the hybrid classical-quantum input as in Eq.~(\ref{eq:realization-qgate}) is replaced  by an arbitrary $\tilde{\cT}$ such that,
    \begin{equation} \label{eq:realization-qgate-noisy}
      \tilde{\cT}_g = \sum_{\mathbf{u} \in \{0, 1\}^l } \cB^{\bu} \otimes \tilde{\cT}^\mathbf{u},
\end{equation}
where for all $\mathbf{u} \in \{0, 1\}^{l} $, $\tilde{\cT}^\mathbf{u}$ is an arbitrary quantum channel with the same input and output space as $\cT^\mathbf{u}$ (in particular, classical registers remain classical), and  $\dnorm{\Tilde{\cT}^\mathbf{u} - \cT^\mathbf{u} } \leq \delta$.  Note that a noisy idle gate realizes a channel $\tilde{\cT}_g$, such that $\dnorm{\tilde{\cT}_g - \cI} \leq \delta$. 
\end{enumerate}
A quantum circuit $\Phi$ with the noise rate $\delta$ realizes a quantum channel in the set,
\begin{multline}
    \trans(\Phi, \delta) := \Big\{ \tilde{\cT}_d \circ \cdots \circ \tilde{\cT}_1 :  \tilde{\cT}_i =  \tr_{\Delta'_i \setminus \Delta_i} \circ \Big[\tilde{\cT}_{g_{i, 1}}[A'_{i1}; A_{i1}] \otimes \cdots 
  \otimes \tilde{\cT}_{g_{i, s_i}}[A'_{is_i}; A_{is_i}] \Big],\\
     \text{with $\tilde{\cT}_{g_{i,j}}$ being a $\delta$-noisy version of $\cT_{g_{i,j}}$ as in Eq.~(\ref{eq:realization-qgate-noisy}). } \Big\}.
\end{multline}
Here, `Trans' in $\trans(\Phi, \delta)$ stands for `Transformation'. 

\subsection{Fault-tolerant state preparation}
One of the ingredients for the proof of our main result in Theorem~\ref{thm:main-cons-over-informal} is a fault-tolerant state preparation method for preparing code states of QLDPC codes. To do this, we use a state preparation method given in~\cite{christandl2024fault}, which can prepare any $n$-qubit state up to an error of weight $O(n\delta)$, under general circuit noise with error rate $\delta$.

\smallskip Before stating the state preparation theorem, we first recall the definition of adversarial channel with noise parameter $\delta$ from~\cite{christandl2024fault}. We also present, for later reference, a result from~\cite{christandl2024fault} that states that any adversarial channel with parameter $\delta$ can be approximated by a superoperator of weight $O(n\delta)$.
\begin{definition}[Adversarial channels] \label{def:adver-can}
   A quantum channel   $\cV: \bL(M^{\otimes n})\to \bL(M^{\otimes n})$ is said to be an adversarial channel with parameter $\delta$ if for any $t > 0$, there exists a superoperator $\cV_t$ of weight $t$, that is, $\cV_t \in \mathscr{E}(n, t) (\cdot) \mathscr{E}(n, t)^*$ (see Def.~\ref{def:weight}), such that 
   \begin{equation} \label{eq:d-v-vt}
    \dnorm{\cV - \cV_t} \leq  \sum_{j > t}^n \binom{n}{j} \delta^j.
\end{equation}
\end{definition}

\begin{lemma} \label{lem:bound-adv}
For any adversarial channel $\cV: \bL(M^{\otimes n})\to \bL(M^{\otimes n})$ with parameter $\delta$, there exists a superoperator $\cV'$ of weight $5n\delta$, such that
\begin{equation} \label{eq:v-prime}
    \dnorm{\cV - \cV'} \leq \exp(- \frac{n\delta}{3})
\end{equation}
\end{lemma}

\begin{proof}
From the definition of adversarial channel, there exists a superoperator $\cV_t$ of weight $t$, such that  
  \begin{align} 
    \dnorm{\cV - \cV_t} &\leq  \sum_{j > t}^n \binom{n}{j} \delta^j \\
    & = (1 + \delta)^n \sum_{j > t}^n \binom{n}{j} \left(\frac{\delta}{1 + \delta}\right)^j \left(\frac{1}{1 + \delta}\right)^{n-j}.
\end{align}
We now consider $n$ independent and identically distributed (i.i.d.) random variables $X_1, X_2, \dots, X_n$, such that $X_i \in \{0, 1\}$, and $\mathrm{Pr}(X_i = 1) = \frac{\delta}{1 + \delta}$, for all $i \in [n]$. We define $X = \sum_i X_i$, then
\begin{equation} \label{eq:ran}
   \mathrm{Pr} ( X > t  ) =  \sum_{j > t}^n \binom{n}{j} \left(\frac{\delta}{1 + \delta}\right)^j \left(\frac{1}{1 + \delta}\right)^{n-j}.
\end{equation}
Therefore, from Eq.~(\ref{eq:t}), Eq.~(\ref{eq:ran}), and $(1 + \delta)^n \leq e^{n\delta}$, we get
\begin{equation} \label{eq:t}
    \dnorm{\cV - \cV_t} \leq   e^{n\delta} \mathrm{Pr} ( X > t  ). 
\end{equation}
For a $\beta > 0$, we take $t = (1 + \beta) \frac{n\delta}{1 + \delta}$. Then, using the multiplicative Chernoff bound, we get 
\begin{equation} \label{eq:mult-chern}
     \mathrm{Pr} ( X > t  ) \leq \exp(-\frac{\beta^2}{2 + \beta} \frac{n\delta}{1 + \delta}) \leq \exp(-\frac{\beta^2 n \delta}{2(2 + \beta)}).
\end{equation}
By taking $\beta = 4$ (i.e., $t \leq  5n\delta$), from Eq.~(\ref{eq:t}) and Eq.~(\ref{eq:mult-chern}), we have that Eq.~(\ref{eq:v-prime}) is satisfied for $\cV' = \cV_t$.
\end{proof}
Below, we present the state preparation theorem from~\cite[Theorem~1]{christandl2024fault}.
\begin{theorem}[State Preparation]\label{thm:state-prep}
    Consider a quantum circuit $\Phi$ with classical input and  $n$-qubit output. 
    Let $\delta$ be a fixed noise rate, considering general circuit noise. Then, for any positive integer $k$, there exists a quantum circuit $\overline{\Phi}$, with the same input and output systems as $\Phi$ and of size $|\Phi| \cdot \mathrm{poly}(k)$,
   such that for any noisy version $\tilde{\cT}_{\overline{\Phi}} \in \trans(\overline{\Phi}, \delta)$,
    \begin{equation} \label{eq:state-prep}
        \dnorm{\tilde{\cT}_{\overline{\Phi}} - \cW \circ \Phi} \leq  O(|\Phi| \sqrt{(c\delta)^k}),
    \end{equation}
where $\cW$ is an adversarial channel with parameter $O(\delta)$, and $c$ is a universal constant.
\end{theorem}
\begin{remark} \label{rem:size-depth-barphi}
The fault-tolerant circuit $\overline{\Phi}$ corresponding to the state preparation circuit $\Phi$ in Theorem~\ref{thm:state-prep} is constructed, using a fault-tolerant scheme~\cite{christandl2024fault, kitaev}, which consists of the following three points, 
\begin{itemize}
    \item[$(a)$] \textbf{Quantum codes:} A sequence of quantum codes $D_k,  k= 1, 2, \dots, $ encoding one logical qubit in $n_k$ physical qubits, with $D_1$ being the trivial code, that is, the corresponding $n_k = 1$.

    \item[$(b)$] \textbf{Logic gate:} For any gate $g$ from a gate set, a sequence of quantum circuits $\Psi_{g, k}, k= 1, 2, \dots$, such that the corresponding noisy version, with noise parameter $\delta > 0$, represents (simulates) the gate $g$ in codes $D_k$, with an accuracy $(c\delta)^k$, for a constant $c > 0$.  

    \item[$(c)$] \textbf{Interface:} A sequence of circuits $\Gamma_{k, l},  k, l= 1, 2, \dots $,  such that the corresponding noisy version with noise parameter $\delta$, represents the identity gate in codes $D_k, D_l$, with accuracy $(c\delta)^{\min\{ k, l \}}$.   
\end{itemize}
Moreover, it is required that size of quantum circuits $\Psi_{g, k}$ and $\Gamma_{k,l}$ scale polynomially, with respect to $k, l$. It has been shown that a fault-tolerant scheme can be built using concatenated codes $\cC_r, r =0, 1, 2, \dots$, where $r$ is the level of concatenation, which is related with the parameter $k$ in Point $(a)$ of the fault-tolerant scheme as $r \sim \log k$~\cite{christandl2024fault}. For $r= 0$, $\cC_r$ is the trivial code and $\cC_1$ is the base code for concatenation. It is required that $\cC_1$ has minimum distance at least $5$, \emph{i.e.}, it corrects at least two errors.

\smallskip The interface circuit $\Gamma_{k, l}$, corresponding to $l = 1$, plays an important role in the construction of the fault-tolerant circuit $\overline{\Phi}$ from Theorem~\ref{thm:state-prep}. Note that $\Gamma_{k, 1}$ maps a logical qubit encoded in $n_k$ physical qubit to a bare logical qubit, while preserving the logical information with a reasonable accuracy $O(\delta)$.
The fault-tolerant circuit $\overline{\Phi}$ corresponding to $\Phi$ and integer $k$ in Theorem~\ref{thm:state-prep} can be taken as follows,
\begin{equation}
\overline{\Phi} := \otimes_{i=1}^n \Gamma_{k, 1} \circ \Phi_k,
\end{equation}
where $\Phi_k$ is obtained from $\Phi$ by replacing each gate $g$ with the corresponding fault-tolerant logic gate $\Psi_{g, k}$. As both $\Psi_{g, k}$ and $\Gamma_{k, l}$ have size scaling polynomially in $k$ and $l$, respectively, it follows that the size of $\overline{\Phi}$ is $|\Phi| \cdot \mathrm{poly}(k)$, and its depth is $d \cdot \mathrm{poly}(k)$, where $d$ is the depth of $\Phi$.
\end{remark}
\section{Constant overhead fault-tolerant quantum computing for general circuit noise} \label{sec:proof}
In this section, we prove our main result, stated in Theorem~\ref{thm:main-cons-over-informal}. We will first formally state our constant overhead result in Theorem~\ref{thm:main-cons-over}, and then provide a proof.

\begin{theorem} \label{thm:main-cons-over}
Let $\Phi$ be a quantum circuit, with classical input and output, working on $x$ qubits and $|\Phi| = \mathrm{poly}(x)$. For any $\epsilon > 0$, and asymptotic overhead $\alpha > 1$, there exists a constant error rate $\delta_{th} > 0$, and a quantum circuit $\Phi_\epsilon, |\Phi_\epsilon| = \mathrm{poly}(x)$, with the same input and output systems as $\Phi$, and working on $x' = \alpha x + o(x)$ qubits, such that for any $\tilde{\cT}_{\Phi_\epsilon} \in \trans(\Phi_\epsilon, \delta)$, $ \delta < \delta_{th}$,
    \begin{equation}
       \| \tilde{\cT}_{\Phi_\epsilon} - \cT_{\Phi} \|_1 \leq  \epsilon.
    \end{equation}       
\end{theorem}

\subsection{Proof of Theorem~\ref{thm:main-cons-over}} 
We consider the fault-tolerant construction from~\cite{gottesman2013fault, fawzi2020constant}, which requires $\Phi$ to be sequential in the following sense: at any layer of the circuit $\Phi$, only one non-trivial gate is applied on a selected set of qubits, and the idle gate is applied on the remaining qubits (see also Remark~\ref{rem:idle}).
 This means we need to map $\Phi$ to a sequential circuit $\Phi_{\mathrm{seq}}$. Any circuit $\Phi$ can be mapped to a sequential circuit $\Phi_{\mathrm{seq}}$ by incurring only a linear overhead in the size of the circuit, $|\Phi_{\mathrm{seq}}| = O(x) |\Phi|$. Note that the linear overhead is due to the additional idle gates, which are accounted in our circuit model. Therefore, without loss of generality, we suppose $\Phi$ to be a sequential circuit.

\smallskip We consider a family of QLDPC codes with rate $R = \frac{2}{1 + \alpha}$, where $\alpha > 1$ is the desired asymptotic overhead and take a code $C$ of type $(n, m)$ in this family (note that asymptotically $n = \tfrac{m}{R}$). For this family of QLDPC codes, we suppose fault-tolerant realizations of error correction and logic gates, as described in the following two paragraphs, respectively.

\paragraph{Fault-tolerant error correction.} Consider $h$ code blocks of the code $C$ and a joint state $\ket{\psi}$ in the code space of these $h$ blocks. Let $\cV$ be an error map on the state $\proj{\psi}$, and suppose that it has  stabilizer reduced weight $sn$ for a constant $ 0 < s < 1$, with respect to each code block (see Defs.~\ref{def:stab-red-weight},~\ref{def:reduced-op}, and~\ref{def:reduced-supop} below for definition of the stabilizer reduced weight). Then, when the error correction circuit with general noise is applied on every code block, it maps $\cV(\proj{\psi})$ to $\cV'(\proj{\psi})$ with high accuracy, where the stabilizer reduced weight of $\cV'$ with respect to each code block is given by $s'n$, with $s' < s$. More precisely, let $\Phi^i_{\mathrm{EC}}$ be the instance of the error correction circuit applied on the block indexed $i \in [h]$, and let $\tilde{\cT}_{\Phi^i_{\mathrm{EC}}} \in \trans(\Phi^i_{\mathrm{EC}}, \delta), i \in [h]$ be a noisy realization under general circuit noise. Then, we have 
\begin{equation} \label{eq:ftec-sketch}
    \| ( \otimes_{i \in [h]} \tilde{\cT}_{\Phi^i_{\mathrm{EC}}}) \circ \cV (\proj{\psi}) - \cV'(\proj{\psi}) \|_1 \leq O( h \exp\left(- O(n \delta) \right) ).
\end{equation}
We provide in Theorem~\ref{thm:cir-noise-ftec}, a fault-tolerant error correction circuit $\Phi_{\mathrm{EC}}$, satisfying Eq.~(\ref{eq:ftec-sketch}). To this end, we show that the single-shot error correction for quantum Tanner codes from~\cite{gu2024single} satisfies Eq.~\eqref{eq:ftec-sketch}, while considering general circuit noise for syndrome extraction. We note that single-shot error correction requires only one round of noisy syndrome measurement, which can be done by a circuit of constant depth (independent of $n$) for QLDPC codes. Moreover, the classical decoding algorithm, estimating the error, is also of constant depth~\cite{gu2024single}. Thus, single-shot error correction is executed in constant time.

\paragraph{Realization of logic gates.} 
The logic gate corresponding to preparing the $\ket{0}$ state is realized using our state preparation procedure given in Theorem~\ref{thm:state-prep}.  For a logic unitary gate, we use gate teleportation as in~\cite{gottesman2013fault, fawzi2020constant}. To realize the logical version of a unitary gate $U_g$ between logical qubits in the same code block, one needs to prepare as ancilla the following entangled code state
\begin{equation} \label{eq:encoded-ancilla-st}
\ket{\psi'_{U_g}} = V \otimes V  \ket{\psi_{U_g}}, 
\end{equation}
where $V$ is the encoding isometry corresponding to the QLDPC code, and $\ket{\psi_{U_g}}$ is a $2m$ qubit entangled state as follows,
\begin{equation} \label{eq:ancilla-gate-tel}
 \ket{\psi_{U_g}} = (\ident \otimes U_g) \left(\frac{\ket{0}\ket{0} + \ket{1}\ket{1}}{\sqrt{2}}\right)^{\otimes m},
 \end{equation}
with $U_g$ acting on a selected set of qubits in the second group of $m$ qubits and identity is applied everywhere else. A similar state is needed when $U_g$ acts on qubits in different blocks. We use Theorem~\ref{thm:state-prep} to prepare the code state $\ket{\psi'_{U_g}}$ fault-tolerantly.

\smallskip During the gate teleportation, one performs logical Bell measurements between the code block, and the first code block of the ancilla state. After the Bell measurement, the state of the remaining ancilla block is given by $U^\dagger_{\mathrm{cor}} U^L_g \ket{\psi}$, where $\ket{\psi}$ is the logical state corresponding to the code block prior to gate teleportation, $U^L_g$ is the logical unitary corresponding to $U_g$ and $U_{\mathrm{cor}}$ is a correction unitary. If $U_g$ is in the Clifford group, then the correction is a Pauli gate which is transversal~\cite[Chapter 6]{grospellier2019constant}. If $U_g$ is in the second level of the Clifford hierarchy (e.g., $U_g$ is the $T$ gate), then the correction is a Clifford gate that we can apply in a transversal way using another state preparation for this correction unitary.

\smallskip We note that the logical Bell measurement can be performed by applying a transversal $\CNOT$ gate between the code blocks, then a transversal Hadamard gate on the first block, and then a transversal single qubit Pauli measurements on the two code blocks~\cite[Chapter 6]{grospellier2019constant}. For a detailed discussion of the implementation of logic gates via ancillary preparation and gate teleportation, we refer to~\cite[Chapter 6]{grospellier2019constant}.

\paragraph{Layout of the circuit $\Phi_\epsilon$.} The quantum circuit $\Phi_\epsilon$ simulating $\Phi$ is constructed as follows: we divide $x$ qubits of $\Phi$ into $x / m$ blocks, each containing $m = O(\sqrt{x})$ qubits\footnote{We have chosen the block size $m = O(\sqrt{x})$ for the simplicity of notation as in~\cite{grospellier2019constant}. In principle, one can choose $m = O(\frac{x}{\mathrm{polylog}(x)})$~\cite{gottesman2013fault}.}. We replace each block of $m$ qubits by $n$-qubits corresponding to a code block of the QLDPC code $C$. The quantum circuit $\Phi_\epsilon$ is obtained by replacing each gate in the circuit $\Phi$ by the corresponding logic gate, which acts within one block or between a pair of blocks, each containing $n$-qubits.

\smallskip As $\Phi$ is a sequential circuit, a layer in $\Phi$ realizes only one non-trivial gate. In $\Phi_\epsilon$, the logical version of a layer in $\Phi$ is realized in several layers at the physical level. The depth of the logical layer is determined by the depth of the non-trivial logic gate applied at that layer. While the non-trivial logic gate is applied on the respective blocks in $\Phi_\epsilon$, repeated error correction is applied on the remaining blocks. We illustrate in Fig.~\ref{fig:layout-sim-circ} the implementation of a logic unitary gate, using gate teleportation, when $U_g$ is applied on a set of logical qubits within one code block.

\smallskip The logic gate corresponding to a unitary $U_g$ is realized in the following two steps (see also Fig.~\ref{fig:layout-sim-circ}): 
\begin{itemize}
    \item[$(i)$] \textbf{Ancilla preparation:} During this step, the required ancilla code state $\ket{\psi'_{U_g}}$ from Eq.~(\ref{eq:encoded-ancilla-st}) is prepared. While the ancilla state is prepared, one keeps applying error correction on the code blocks of the QLDPC codes (see Fig.~\ref{fig:layout-sim-circ}). We know that there exists a circuit $\Psi'_{U_g}$, working on $O(n)$ qubits, which prepares $\ket{\psi'_{U_g}}$, with size $|\Psi'_{U_g}| = O(\mathrm{poly}(n))$ and  depth $d' = O(\mathrm{poly}(n))$~\cite[Chapter 4]{gottesman-thesis}\footnote{It is possible to reduce the depth of $\Psi'_{U_g}$ to $O(\log n)$ at the cost of extra ancillary qubits~\cite{moore2001parallel}.}. To prepare $\ket{\psi'_{U_g}}$, we consider the fault-tolerant circuit $\overline{\Psi'}_{U_g}$ associated with $\Psi'_{U_g}$, according to Theorem~\ref{thm:state-prep}.   Let $\rho$ be the output of $\overline{\Psi'}_{U_g}$, with general circuit noise. Then, from Eq.~(\ref{eq:state-prep}), we have   
\begin{equation} \label{eq:err-ancilla-prep}
    \|\rho -  \cV(\proj{\psi'_{U_g}})  \|_1 \leq O(|\Psi'_{U_g}|) (c \delta)^{k/2},
\end{equation}
where $\cV$ is a superoperator of weight $O(n \delta)$ and $k$ is the parameter from Theorem~\ref{thm:state-prep}. Therefore, $\overline{\Psi'}_{U_g}$ prepares the code state $\ket{\psi'_{U_g}}$ up to an error of weight $O(n\delta)$, with very high accuracy given that $k$ is sufficiently large. From Theorem~\ref{thm:state-prep} and Remark~\ref{rem:size-depth-barphi}, note that the size and the depth of $\overline{\Psi'}_{U_g}$ are given by,
\begin{align}
  |\overline{\Psi'}_{U_g}| = |\Psi'_{U_g}| \: \mathrm{poly}(k) = O(\mathrm{poly}(n) \: \mathrm{poly}(k)) \label{eq:size-ft-enc} \\
\overline{d'} = d' \: O(\mathrm{poly}(k)) = O(\mathrm{poly}(n) \: \mathrm{poly}(k))  \label{eq:depth-ft-enc}
\end{align}
\item[$(ii)$] \textbf{Gate teleportation:} The prepared ancilla state is then used to implement the logic gate on the selected code block, using gate teleportation. On the remaining blocks, error correction circuits are applied. 
\end{itemize}

From Eq.~\eqref{eq:depth-ft-enc}, we note that the time-complexity of the ancillary state preparation increases with $n$ in contrast with the single-shot error correction on QLDPC code. This, however, does not pose a problem since while the ancillary state is being prepared, repeated rounds of error correction are performed on the QLDPC code (see Fig.~\ref{fig:layout-sim-circ}), thereby logical information remains protected throughout the state preparation.
\paragraph{Error bound on computation.} Our state preparation procedure together with the transversality of the gates, used in gate teleportation, imply that only an error of weight $c_g  \delta n$, for a constant $c_g > 0$, is introduced during the implementation of the logic gate. By choosing the parameter $\delta$ small enough so that $s' + c_g \delta < s$, we can make sure that error can be suppressed in the next round of error correction, as given in Eq.~(\ref{eq:ftec-sketch}). Therefore, after the $t^{th}$ logical layer of the quantum circuit  $\Phi_\epsilon$, with general noise (supposing the last layer is error correction), we obtain a quantum state $\rho_t$, which is close to the desired code state $\proj{\psi_t}$ (i.e, the state at the $t$-th step of $\Phi_\epsilon$ without noise) in the following sense, 
\begin{equation} \label{eq:comp-err}
    \|\rho_t - \cV'(\proj{\psi_t}) \|_1 \leq O \left( t \overline{d'}\sqrt{x}  \exp\left(- O(\sqrt{x}\delta) \right) \right) + O( t |\Psi_{U_g}|) (c \delta)^{k/2},
\end{equation}
where $\cV'$ is a superoperator of weight $s'n$, with respect to each code block. The first error term on the right hand side of Eq.~(\ref{eq:comp-err}) is due to fault-tolerant error correction (see Eq.~(\ref{eq:ftec-sketch})), and the second error term is due to the ancilla state preparation (see Eq.~(\ref{eq:err-ancilla-prep})). For the second term, using $|\Psi_{U_g}| = O(\mathrm{poly}(n)) = O(\mathrm{poly}(\sqrt{x}))$, and  $t = \mathrm{poly}(x)$,  it suffices to choose $k = \mathrm{polylog}(\frac{x}{\epsilon})$ to make this term of order $\epsilon$. Therefore, from Eq.~(\ref{eq:size-ft-enc}) and Eq.~(\ref{eq:depth-ft-enc}), the size and depth of the state preparation circuit is given by,
\begin{align}
  |\overline{\Psi'}_{U_g}|  &=  O \left (\mathrm{poly}(x) \:  \mathrm{polylog}\left(\frac{1}{\epsilon}\right) \right). \\
  \overline{d'} &= O \left (\mathrm{poly}(x) \:  \mathrm{polylog}\left(\frac{1}{\epsilon}\right) \right). \label{eq:depth-ancilla-circ}
\end{align}
Finally, since $t = \mathrm{poly}(x)$, the first term approaches to zero as $x \to \infty$.

\paragraph{Proof of constant qubit overhead.} We now count the number of qubits in the quantum circuit $\Phi_\epsilon$,

\begin{itemize}
    \item [(i)] \textbf{Data qubits:} The number of data qubits corresponding to all the code blocks is given by $\frac{x}{R}$.

    \item [(ii)] \textbf{Ancilla qubits for generator measurements:} The number of ancilla qubits used for error correction in each block is $n - m$. Therefore, the total number of ancilla qubits needed for error correction is given by, $\frac{(n-m)x}{m} = x (\frac{1}{R}-1)$. 

    \item [(iii)] \textbf{Ancilla qubits for state preparation:}   As discussed before, we can take $k = \mathrm{polylog}(\frac{x}{\epsilon})$ to achieve an error of order $\epsilon$ in ancilla state preparation.   Since the qubit overhead of preparation according to Theorem~\ref{thm:state-prep} is $\mathrm{poly}(k)$, the number of qubits needed to prepare the ancilla state is given by
    \begin{equation}
              O\left(n \: \mathrm{polylog}\left(\frac{x}{\epsilon}\right)\right)=    O\left(\sqrt{x} \: \mathrm{polylog} \left(\frac{x}{\epsilon}\right)\right).
    \end{equation}
\end{itemize}
Therefore, the total number of qubits is given by,
\begin{equation}
   x' =  2x \frac{1}{R} - x +  O\left(\sqrt{x} \: \mathrm{polylog} \left(\frac{x}{\epsilon}\right)\right) = \alpha x + o(x).  
\end{equation}
\section{Single-shot error correction for QLDPC codes against general errors} \label{sec:single-shot-err}

For fault-tolerant error correction, one needs to consider the errors on both data qubits and syndrome bits as errors also happen during the syndrome extraction procedure. Usually, syndrome errors are dealt with by extracting several rounds of syndromes~\cite{shor1996fault, fowler2012surface}. Single-shot error correction uses only a single round of noisy syndrome extraction to perform error correction~\cite{bombin2015single, campbell2019theory}. Recently, it has been shown that the (asymptotically) good QLDPC codes of~\cite{leverrier2022quantum}, namely, quantum Tanner codes, support a single-shot decoding~\cite{gu2024single}, where Pauli errors are applied on the data qubits, and bit flip errors on the syndrome bits. In this section, we generalize the single shot error correction for general linear error operators on data qubits. This generalization will be needed for the fault-tolerant error correction under  general circuit noise, given in Theorem~\ref{thm:cir-noise-ftec} in the next section.

\medskip We below provide a definition of QLDPC codes.
\begin{definition}[$(r, s)$-QLDPC codes~\cite{gottesman2013fault}] \label{def:ldpc-code}
A stabilizer code is said to be an $(r, s)$-QLDPC code if there exists a choice of stabilizer generators $\{ M_1, \dots, M_l \}$ such that the weight of $M_i$ is at most $r$, for all $i \in [l]$, and each qubit is non-trivially involved in at most $s$ generator measurements. 
\end{definition}

\medskip In what follows, we fix a Calderbank-Shor-Steane (CSS) QLDPC code $C$ of type $(n, m)$, i.e,  encoding $m$ logical qubits in $n$ physical qubits, and with minimum distance $d_{\mathrm{min}}$. Moreover, we suppose that $C$ admits a single-shot decoding procedure. For example, one can choose a quantum Tanner code from~\cite{leverrier2022quantum} and the decoder analysis of~\cite{gu2024single}. We denote the  generating set $\{ M_1, \dots, M_l \}, l < n$, as in Def.~\ref{def:ldpc-code} (generally, $l = n - m$).  The encoding channel of $C$ is given by $V(\cdot)V^\dagger$, where $V: (\mathbb{C}^2)^{\otimes m} \to (\mathbb{C}^2)^{\otimes n}$ is an isometry, and the code space $\overline{L} := \mathrm{Im}(V) \subseteq (\mathbb{C}^2)^{\otimes n}$. 

\smallskip The error correction map is given by a classical controlled unitary 
\[ \cU_{\mathrm{ec}}: \mathbb{B}^{\otimes l} \otimes \bL((\mathbb{C}^2)^{\otimes n}) \to \mathbb{B}^{\otimes l} \otimes \bL((\mathbb{C}^2)^{\otimes n}),\] 
where $\mathbb{B}^{\otimes l}$ refers to the syndrome bits, that is, the outcome of generator measurements $\{ M_1, \dots, M_l \}$. For stabilizer codes, $\cU_{\mathrm{ec}}$ is a controlled Pauli, where $l$ bits act as control and $n$-qubits as target. For simplicity of notation, we define $\cU_{ec} := U_{ec} (\cdot) U_{ec}^\dagger$, where $U_{ec}$ is a controlled Pauli, with $l$ qubit control and $n$-qubit target.

\medskip In the following definition, we recall the stabilizer reduced error weight.

\begin{definition}[Stabilizer reduced error weight~\cite{gu2024single}] \label{def:stab-red-weight}
Consider a CSS code of type $(n, m)$, and let $X^{\be_x}, \be_x  \in \{0, 1\}^{n}$ be a Pauli-$X$ error. Let $|\be_x|$ denote the Hamming weight of $\be_x$. The stabilizer reduced weight $|\be_x|_\mathrm{red}$ of $\be_x$ is defined as, 
\begin{equation*} 
   |\be_x|_\mathrm{red} := \min \{ |\be'_x| : X^{\be'_x} \text{ is  equivalent to $X^{\be_x}$ up to stabilizer multiplication} \}. 
\end{equation*}
 Similarly, we can define stabilizer reduced weight for Pauli-$Z$ errors.  The stabilizer reduced weight of the total error $\be = (\be_x, \be_z)$ is defined as $|\be|_\mathrm{red} := \max \{|\be_x|_\mathrm{red}, |\be_z|_\mathrm{red}\}$. 
\end{definition}

Def.~\ref{def:reduced-op} and Def.~\ref{def:reduced-supop} extend the definition of the reduced error weight to linear error operators and superoperators, respectively. To this end, we consider error operators that act jointly on the $n$ physical qubits of the code and on a $\overline{n}$-qubit auxiliary system. This is done keeping in mind that we will consider error correction of multiple code blocks in parallel; therefore, we would need a notion of reduced error weight with respect to each code block.
\begin{definition} \label{def:reduced-op} 
    Consider a CSS code of type $(n, m)$. Let $E \in \bL((\mathbb{C}^2)^{\otimes \overline{n}} \otimes (\mathbb{C}^2)^{\otimes n})$ be an operator, jointly acting on a code block containing $n$ physical qubits and on an auxiliary system containing $\overline{n}$-qubits.  We expand $E$ as,
\begin{equation} \label{eq:expansion-k}
    E = \sum_{\be_x, \be_z \in \{0, 1\}^n}    E_{\be_x, \be_z} \otimes X^{\be_x} Z^{\be_z},
\end{equation}
where $E_{\be_x, \be_z} := \frac{1}{2^n} \tr_{[n]}(E (\ident_{\overline{n}} \otimes Z^{\be_z} X^{\be_x}) ) \in  \bL((\mathbb{C}^2)^{\otimes \overline{n}}) $, where $\tr_{[n]}$ refers to tracing out the $n$-qubits in the code block. The reduced weight of $E$ with respect to the code block is defined as 
\begin{equation*} 
    |E|_\mathrm{red} := \max \{ |\be|_\mathrm{red}, \be = (\be_x, \be_z) \in \{0, 1\}^{2n} : E_{\be_x, \be_z} \neq 0\}.
\end{equation*}
\end{definition}
\medskip
\begin{definition} \label{def:reduced-supop}
For a superoperator $\cV: \bL((\mathbb{C}^2)^{\otimes \overline{n}} \otimes (\mathbb{C}^2)^{\otimes n}) \to \bL((\mathbb{C}^2)^{\otimes \overline{n}} \otimes (\mathbb{C}^2)^{\otimes n})$, we say $|\cV|_\mathrm{red} \leq t$ if there exists a decomposition,
\begin{equation}
    \cV = \sum_{i} E_i (\cdot) E_i^{\prime \dagger},
\end{equation}
such that $\max_i \{ |E_i|_\mathrm{red}, |E'_i|_\mathrm{red} \} \leq t$.
\end{definition}

\subsection{Results for single-shot error correction}
Let $\Pi_{\sigma} \in \bL((\mathbb{C}^2)^{\otimes n})$ be the projector on the space corresponding to the syndrome $\sigma := (\sigma_1, \dots, \sigma_l) \in \{0, 1\}^l$, that is,
\begin{equation}
    \Pi_\sigma = \prod_{i \in [l]} \frac{1}{2} (1 + (-1)^{\sigma_i} M_i).
\end{equation}
Let $\sigma_\be = (\sigma_{\be_x}, \sigma_{\be_z})$ be the syndrome corresponding to the Pauli error given by $\be = (\be_x, \be_z)$, that is, $\sigma_{\be_x} = H_Z \be_x, \sigma_{\be_z} = H_X \be_z$, where $H_Z$, and $H_X$ are $Z$ and $X$ type parity check matrices of the CSS code $C$, respectively. We have for any $\ket{\psi} \in \overline{L}$,
\begin{equation} \label{eq:pauli-proj}
    \Pi_{\sigma} X^{\be_x} Z^{\be_z} \ket{\psi} = \alpha_{\sigma, \sigma_\be} X^{\be_x} Z^{\be_z} \ket{\psi},
\end{equation}
where $\alpha_{\sigma, \sigma_\be} = 1$ if $\sigma = \sigma_\be$ and  $\alpha_{\sigma, \sigma_\be} = 0$ if $\sigma \neq \sigma_\be$. The quantum channel corresponding to the (perfect) syndrome measurement is given by 
\begin{equation}
    \cM_{\mathrm{syn}}(\cdot) :=  \sum_{\sigma \in \{0, 1\}^l} \proj{\sigma} \otimes \Pi_\sigma (\cdot) \Pi_\sigma
\end{equation}
The following channel captures the effect of the syndrome error $\be^{\mathrm{syn}} \in \{0, 1\}^l$, where the recorded syndrome is $\sigma$, but the realized projective measurement is $\Pi_{\sigma \oplus \be^{\mathrm{syn}}}$,
\begin{equation}  \cM^{\be^{\mathrm{syn}}}_{\mathrm{syn}}(\cdot) :=  \sum_{\sigma \in \{0, 1\}^l} \proj{\sigma} \otimes \Pi_{\sigma \oplus \be^{\mathrm{syn}}} (\cdot) \Pi_{\sigma \oplus \be^{\mathrm{syn}}}.
\end{equation}

\smallskip Below, we present the one shot-decoding theorem from~\cite[Theorem 3.5]{gu2024single}, stated in a slightly different language. We have taken constant $\alpha = 0$ in ~\cite[Theorem 3.5]{gu2024single}. Informally, the theorem states that if the data error and the syndrome errors are bounded, then the error correction procedure suppresses the input error, even though it may not completely remove it (see Fig.~\ref{fig:single-shot-dec} below). 
\begin{figure}[!t]
\centering
\begin{subfigure}[b]{0.7 \linewidth}
\centering 
\resizebox{1.0\textwidth}{!}{
\begin{quantikz}
 & \qwbundle{n} & \qw    & \gate[1, style = {inner
ysep=4pt}]{X^{\mathbf{e}_x}Z^{\mathbf{e}_z}} \gategroup[1,steps=1,style={dashed,rounded
corners,fill=blue!20, inner
xsep=2pt},background,label style={label
position= above,anchor=north,yshift=0.35cm}]{ $\mathbf{e}_{|\mathrm{red}|} \leq \beta Sn + Rn$} &\gate[2]{\Phi_{\mathrm{syn}}} & & & \gate[2]{\text{Corr}} & &   \\
\lstick{ $\ket{0}^{\otimes l}$} & \qwbundle{l} & \qw & \qw &  & \gate[1, style = {inner
ysep=4pt}]{\mathbf{e}^{\mathrm{syn}}} \setwiretype{c} \gategroup[1,steps=1,style={dashed,rounded
corners,fill=blue!20, inner
xsep=2pt},background,label style={label
position= below,anchor=north,yshift=-0.2cm}]{$|\mathbf{e}^{\mathrm{syn}}| \leq Sn$}  & \gate[1, style = {inner ysep = 4pt}]{\mathrm{Dec}} & & \trash{\mathrm{trace}} 
\end{quantikz}
}
\end{subfigure}
\text{$=$}
\begin{subfigure}[b]{0.25 \linewidth}
\centering
\begin{quantikz}
 & \qwbundle{n} & \gate[1, style = {inner
ysep=4pt}]{X^{\mathbf{e}'_x}Z^{\mathbf{e}'_z}} \gategroup[1,steps=1,style={dashed,rounded
corners,fill=blue!20, inner
xsep=2pt},background,label style={label
position= above,anchor=north,yshift=0.4cm}]{ $\mathbf{e}'_{|\mathrm{red}|} \leq \beta Sn$}
\end{quantikz}
\end{subfigure}
\caption{This figure illustrates the single-shot error correction procedure of~\cite{gu2024single}. The input is a code state of a QLDPC code, subjected to a Pauli error $X^{\be_x} Z^{\be_z}$. Then, the noiseless syndrome extraction circuit is applied, after which a bit-flip error $\be^{\mathrm{syn}}$ affects the syndrome bits. Following classical decoding and error correction, the remaining error on the code state is $X^{\be'_x} Z^{\be'_z}$. There exist constants $R > 0, S >0, \beta >0$ with $\beta Sn + Rn < \frac{d_{min}}{2}$, such that if the stabilizer reduced weight of the data error satisfies $|\be|_{\mathrm{red}} \leq \beta S n + R n$ and the syndrome error satisfies $|\be^{\mathrm{syn}}| \leq S n$, then the residual error satisfies $|\be'|_{\mathrm{red}} \leq \beta S n$. Thus, the logical error is effectively suppressed using single-shot error correction.}
\label{fig:single-shot-dec}
\end{figure}
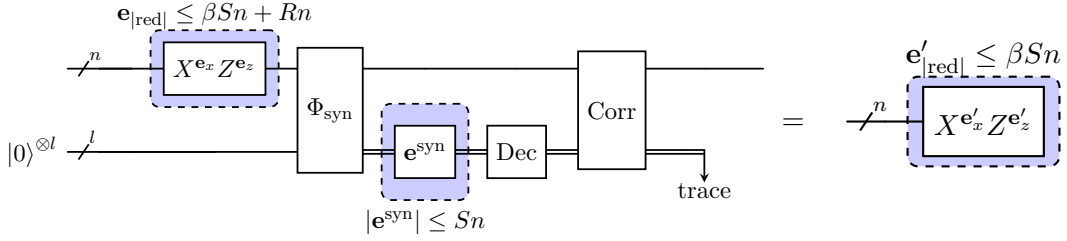

\begin{theorem}[Single-shot decoding~\cite{gu2024single}] \label{thm:single-shot-dec}
There exist constants $R >0, S >0, \beta > 0$ with $\beta Sn + Rn < \frac{d_{min}}{2}$, such that the following holds:

Let $\ket{\psi} \in \overline{L}$ be a quantum state in the code space of the QLDPC code $C$. Suppose the data error is $\be = (\be_x, \be_z) \in \{0, 1\}^{2n}$ with reduced weight $|\be|_{\mathrm{red}} \leq \beta S n + Rn$ and the syndrome error is $\be^{\mathrm{syn}} \in \{0,1\}^l$ with weight $|\be^{\mathrm{syn}}| \leq Sn$. Then,
\begin{equation} \label{eq:ft-syn-ec}
U_{\mathrm{ec}} \left( \ket{\sigma} \otimes \left( \Pi_{\sigma \oplus \be^{\mathrm{syn}}} X^{\be_x} Z^{\be_z} \ket{\psi} \right) \right)
= \alpha_{\sigma + \be^{\mathrm{syn}}, \sigma_\be} \ket{\sigma} \otimes X^{\be'_x} Z^{\be'_z} \ket{\psi},
\end{equation}
where $\sigma \in \{0,1\}^l$, the residual error $(\be'_x, \be'_z)$ satisfies $|\be'|_{\mathrm{red}} \leq \beta S n$, and the factor $\alpha_{\sigma + \be^{\mathrm{syn}}, \sigma_\be}$ is defined as in Eq.~\eqref{eq:pauli-proj}.
\end{theorem}
Note that Eq.~(\ref{eq:ft-syn-ec}) can be stated using the encoding isometery $V: (\mathbb{C}^2)^{\otimes m} \to (\mathbb{C}^2)^{\otimes n}$
   \begin{equation} \label{eq:ft-syn-ec-1}
       U_{ec} \left(\ket{\sigma} \otimes (\Pi_{\sigma \oplus \be^{\mathrm{syn}}} X^{\be_x} Z^{\be_z} V) \right) = \alpha_{\sigma + \be^{\mathrm{syn}}, \sigma_\be} \ket{\sigma} \otimes (X^{\be'_x} Z^{\be'_z}V) ,
   \end{equation}

\medskip In the following corollary, we show that Theorem~\ref{thm:single-shot-dec} can be generalized for arbitrary linear error operators.

\begin{corollary}[Single-Shot decoding for non-Pauli errors] \label{cor:single-shot-dec}
There exist constants $R > 0, S > 0, \beta > 0$ with $\beta Sn + Rn < \frac{d_{\min}}{2}$, such that the following holds: 

Let $\ket{\psi} \in (\mathbb{C}^2)^{\otimes \overline{n}} \otimes \overline{L}$ be a quantum state on the joint space of the code space $\overline{L}$ and a $\overline{n}$-qubit auxiliary system.  Suppose the data error is $E \in \bL((\mathbb{C}^2)^{\otimes \overline{n}} \otimes (\mathbb{C}^2)^{\otimes n})$, with the reduced weight
\begin{equation*}
 |E|_\mathrm{red} \leq \beta S n + Rn < \tfrac{d_{\min}}{2},   
\end{equation*}
 and the syndrome error is  $\be^{\mathrm{syn}} \in \{0, 1\}^{l}$,  $|\be^{\mathrm{syn}}| \leq Sn$. Then,
   \begin{equation}
      (\ident_{\overline{n}} \otimes U_{ec}) (\ket{\sigma} \otimes (\ident_{\overline{n}} \otimes \Pi_{\sigma \oplus \be^{\mathrm{syn}}}) E \ket{\psi})  =  \ket{\sigma}  \otimes (E_{\overline{n}} \otimes X^{\be'_x} Z^{\be'_z}) \ket{\psi},
   \end{equation}
   where $\sigma \in \{0, 1\}^{l}$, $E_{\overline{n}}  \in  \bL((\mathbb{C}^2)^{\otimes \overline{n}}) $ and for $\be' = (\be'_x, \be'_z)$, $|\be'|_\mathrm{red} \leq \beta S n, $. \end{corollary}
\begin{proof}
Using $\be = (\be_x, \be_z)$, we have 
\begin{align} 
    (\ident_{\overline{n}} \otimes \Pi_{\sigma} ) E \ket{\psi} &= \sum_{\be}   (E_{\be} \otimes \Pi_\sigma X^{\be_x} Z^{\be_z}) \ket{\psi} \nonumber \\
     &= \sum_{\be : \sigma_\be = \sigma}  (E_{\be} \otimes \Pi_\sigma X^{\be_x} Z^{\be_z} ) \ket{\psi} \label{eq:k-proj},
\end{align}
where the first equality uses the expansion of $E$ according to Eq.~(\ref{eq:expansion-k}), and the second equality uses Eq.~\eqref{eq:pauli-proj}. 
Note that two errors with the same syndrome are separated by either a stabilizer generator or a logical operator.  The reduced error weight of a logical operator is at least the minimum distance $d_{\min}$ of the code. Since $|E|_\mathrm{red} < d_{\min}/2$, all the errors $\{\be: \sigma_\be = \sigma\}$ in Eq.~(\ref{eq:k-proj}) are stabilizer equivalent to each other. Therefore, Eq.~(\ref{eq:k-proj}) reduces to  
\begin{equation} \label{eq:k-proj-1}
     (\ident_{\overline{n}} \otimes \Pi_{\sigma} ) E\ket{\psi} =   (   E_{\overline{n}, \sigma} \otimes \Pi_\sigma X^{\be_x} Z^{\be_z}) \ket{\psi},
\end{equation}
where $\be = (\be_x, \be_z), |\be|_\mathrm{red} \leq \beta Sn + Rn$ is such that $\sigma_\be = \sigma$, and $E_{\overline{n}, \sigma} := \sum_{\be : \sigma_\be = \sigma} E_\be$. We now have 
\begin{multline}
      (\ident_{\overline{n}} \otimes U_{ec}) (\ket{\sigma} \otimes (\ident_{\overline{n}} \otimes \Pi_{\sigma \oplus \be^{\mathrm{syn}}}) E \ket{\psi})  =  (\ident_{\overline{n}} \otimes U_{ec})(\ket{\sigma} \otimes (E_{\overline{n}, \sigma + \be^{\mathrm{syn}}} \otimes \Pi_{\sigma \oplus \be^{\mathrm{syn}}} X^{\be_x} Z^{\be_z}) \ket{\psi} ) \nonumber \\  
      =\ket{\sigma} \otimes (E_{\overline{n}, \sigma + \be^{\mathrm{syn}}} \otimes X^{\be'_x} Z^{\be'_z} ) \ket{\psi} \quad \quad \quad \quad \quad \quad ,
\end{multline}
where in the first equality, the syndrome of $\be = (\be_x, \be_z)$ is $\sigma_{\be} = \sigma \oplus \be^{\mathrm{syn}}$ from Eq.~(\ref{eq:k-proj-1}) (note that $U_{ec}$ acts on the syndrome state and the code block) and in the second equality, $|\be'|_\mathrm{red} \leq \beta S n$, using Theorem~\ref{thm:single-shot-dec} (see also Eq.~(\ref{eq:ft-syn-ec})).
\end{proof}
\subsection{Single-shot error correction against general circuit noise}
In this section, we consider the syndrome extraction circuit of a QLDPC code under general circuit noise with parameter $\delta$. We show that single shot error correction, i.e., the error correction based on a single round of noisy syndrome, suppresses the input error on any code state, given that $\delta$ is sufficiently small (see Theorem~\ref{thm:cir-noise-ftec} below). 

\smallskip We know that for QLDPC codes, there exists a syndrome extraction circuit of constant (i.e., independent of $n$) depth~\cite{gottesman2013fault}. Throughout this section, we denote by $\Phi$ the syndrome extraction circuit. Let $d = O(1)$ be the depth of $\Phi$.  The syndrome extraction circuit $\Phi$ contains the following steps (see also Fig.~\ref{fig:syn-ext-circuit}),
\begin{itemize}
\item[1] \textbf{Initialization}: The input to the circuit $\Phi$ corresponds to the $n$ data (physical)  qubits of the QLDPC code. In the first layer, $i = 1$, of the quantum circuit, we initialize $l$ ancilla qubits (one for each generator measurement) in the $\ket{0}$ or $\ket{+}$ state.

\item[2] \textbf{Entanglement}: In layers, $i = 2, \dots, d - 1$, we entangle the data qubits with ancilla qubits by applying $\CNOT$ gates. 

\item[3] \textbf{Measurement}: In the last layer, $i = d$, the ancilla qubits are measured in the Pauli-$Z$ or $X$ basis. 
\end{itemize}

\begin{figure}[!t]
\centering
\begin{tikzpicture}

\tikzset{
  half circle/.style={
      semicircle,
      shape border rotate=90,
      anchor=chord center
      }
}

\tikzset{meter/.append style={draw, inner sep=5, rectangle, font=\vphantom{A}, minimum width=20, line width=.8,
 path picture={\draw[black] ([shift={(.1,.3)}]path picture bounding box.south west) to[bend left=50] ([shift={(-.1,.3)}]path picture bounding box.south east);\draw[black,-latex] ([shift={(0,.1)}]path picture bounding box.south) -- ([shift={(.3,-.1)}]path picture bounding box.north);}}}
 
\draw 
(1.0,0) node[inner xsep=0.3cm, inner ysep=2.5cm, draw, fill=blue!20] (a) {${\cT_2}$}
(3.0,0) node[inner xsep=0.3cm, inner ysep=2.5cm, draw, fill=blue!20] (f) {${\cT_3}$}
(7.5,0) node[inner xsep=0.3cm, inner ysep=2.5cm, draw, fill=blue!20] (g) {${\cT_{d-1}}$}
;

\draw 
(-0.6, 2.3) node[draw, minimum width = 20, inner sep = 5, rectangle, fill=blue!20](b){$\cI$}  
(-0.3, -2.3) node[draw, half circle, minimum size = 3mm, fill=blue!20](c){}
(-0.6, -0.2) node[draw, minimum width = 20, inner sep = 5, rectangle, fill=blue!20](d){$\cI$}
(-0.3, -1.0) node[draw, half circle, minimum size = 3mm, fill=blue!20](e){} 
;

\draw
($0.5*(b.south) + 0.5*(d.north)$) node[rotate = 90] (y) {$\cdots$}
(y) ++(0, 0.6) node[rotate = 90] () 
{$\cdots$}
(y) ++(0, -0.6) node[rotate = 90] () 
{$\cdots$}
($0.5*(c.south) + 0.5*(e.north)$) node[rotate = 90] () {$\cdots$} 
;

\draw
(d.west) to ++(-0.5, 0)
(b.west) to ++(-0.5, 0)
;

\draw
(b) to (b -| a.west)
(c) to (c -| a.west)
(d) to (d -| a.west)
(e) to (e -| a.west)
;

\draw
(b -| a.east) to (b -| f.west)
(c -| a.east) to (c -| f.west)
(d -| a.east) to (d -| f.west)
(e -| a.east) to (e -| f.west)
;

\draw
(b -| f.east) to ++(0.5, 0)
(c -| f.east) to ++(0.5, 0)
(d -| f.east) to ++(0.5, 0)
(e -| f.east) to ++(0.5, 0)
;

\draw
(b -| g.west) to ++(-0.5, 0)
(c -| g.west) to ++(-0.5, 0)
(d -| g.west) to ++(-0.5, 0)
(e -| g.west) to ++(-0.5, 0)
;

\draw
($0.5*(b -| f.east) + 0.5*(b -| g.west)$) node[] (x) {$\cdots$}
(x) ++(0.7, 0) node[] () 
{$\cdots$}
(x) ++(-0.7, 0) node[] () 
{$\cdots$}
;

\draw
($0.5*(c -| f.east) + 0.5*(c -| g.west)$) node[] (x1) {$\cdots$}
(x1) ++(0.7, 0) node[] () 
{$\cdots$}
(x1) ++(-0.7, 0) node[] () 
{$\cdots$}
;

\draw
($0.5*(d -| f.east) + 0.5*(d -| g.west)$) node[] (x2) {$\cdots$}
(x2) ++(0.7, 0) node[] () 
{$\cdots$}
(x2) ++(-0.7, 0) node[] () 
{$\cdots$}
;

\draw
($0.5*(e -| f.east) + 0.5*(e -| g.west)$) node[] (x3) {$\cdots$}
(x3) ++(0.7, 0) node[] () 
{$\cdots$}
(x3) ++(-0.7, 0) node[] () 
{$\cdots$}
;

\draw
(b -| g.east) ++(1.4, 0) node[draw, minimum width = 20, inner sep = 5, rectangle, fill=blue!20](z1){$\cI$}  
(d -| g.east) ++(1.4, 0) node[draw, minimum width = 20, inner sep = 5, rectangle, fill=blue!20](z2){$\cI$}  
;

\draw
(b -| g.east) to (z1.west) 
(z1.east) to ++(0.5, 0)
(d -| g.east) to (z2.west) 
(z2.east) to ++(0.5, 0)
;

\draw
(e -| g.east) ++(1.4, 0) node[meter, fill=blue!20](h){}
(c -| g.east) ++(1.4, 0) node[meter, fill=blue!20](i){}
;

\draw
(e -| g.east) to (h.west) 
(c -| g.east) to (i.west) 
;

\draw[double]
(h.east) to ++(0.4, 0)
(i.east) to ++(0.4, 0)
;

\draw
($0.5*(z1.south) + 0.5*(z2.north)$) node[rotate = 90] (y) {$\cdots$}
(y) ++(0, 0.6) node[rotate = 90] () 
{$\cdots$}
(y) ++(0, -0.6) node[rotate = 90] () 
{$\cdots$}
($0.5*(h.south) + 0.5*(i.north)$) node[rotate = 90] () {$\cdots$} 
;

\draw[decorate, decoration = brace]
($(d) + (-1.0, -0.1)$) -- ($(b) + (-1.0, 0.1)$) 
 ;

\draw[decorate, decoration = brace]
($(c) + (-1.0, -0.3)$) -- ($(e) + (-1.0, 0.3)$) 
 ;

 \draw
($(d) + (-1.2, 1.2)$)  node[rotate = 90](){Data qubits}
;

 \draw
($(c) + (-1.2, 0.6)$)  node[rotate = 90](){Ancilla qubits}
;
\end{tikzpicture} 
\caption{This figure illustrates the channel $\cT_\Phi$ realized by the syndrome extraction circuit $\Phi$. In the first layer, $i = 1$, the ancilla qubits are initialized and in the last layer, $i = d$, they are measured. In layers, $i = 2, \dots, d-1$, the data and ancilla qubits are entangled, using two qubit gates.}
\label{fig:syn-ext-circuit}
\end{figure}
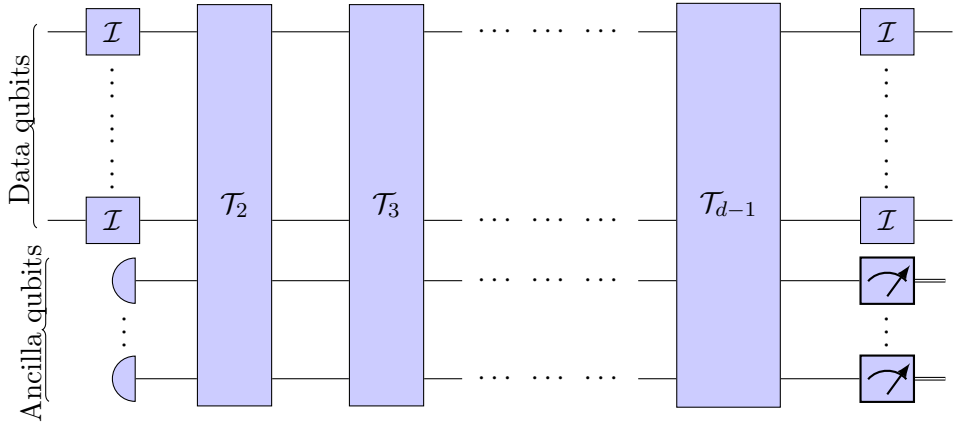
Note that for a QLDPC code, with at least one $X$ and $Z$ type generator of weight bigger than or equal to $2$, it holds $d \geq 4$. We denote the set of data qubits by $D$, and the set of ancilla qubits by $A$. The quantum channel realized by the circuit $\Phi$ is given by,
\begin{equation} \label{eq:syn}
    \cT_{\Phi} := \cT_d \circ \cdots \circ \cT_2 \circ \cT_1,
\end{equation}
where the quantum channel realized by the $i^{th}$ layer is given by
\[ \cT_i := \otimes_{j = 1}^{t_i} \cT_{ij}, i \in [d], \]
with $\cT_{ij}$ being the channel corresponding to the location $(i, j)$, and  $t_i$ being the number of locations (including idle gates) at the $i^{th}$ layer (see also Def.~\ref{def:quantum_circuit}). Note that the following holds,
\begin{equation} \label{eq:t_i-bound}
 n \leq t_i < 2n, \forall i \in [d].    
\end{equation}
In Eq.~(\ref{eq:t_i-bound}), the upper bound follows using that the number of locations at any layer could not exceed $n + l <2n$ (recall that $l < n$ is the number of generator measurements). The lower bound follows from the observation that the maximum number of two qubit locations at any layer is $l$, since each such location acts on a data and an ancilla qubit.

\smallskip Applying general circuit noise with parameter $\delta$ on $\Phi$, we get the set of quantum channels $\trans(\Phi, \delta)$  (see also Section~\ref{sec:noise-qcirc}). Any $\tilde{\cT}_{\Phi} \in \trans(\Phi, \delta)$ can be written as,
\begin{equation} \label{eq:noisy-syn}
    \tilde{\cT}_\Phi = \tilde{\cT}_d \circ \cdots \circ \tilde{\cT}_2 \circ \tilde{\cT}_1,
\end{equation}
where 
\[
  \tilde{\cT}_i = \otimes_{j = 1}^{t_i} \tilde{\cT}_{ij},  
\]
%
with $\tilde{\cT}_{ij}$ such that $\dnorm{\tilde{\cT}_{ij} - \cT_{ij}} \leq \delta$. 

\smallskip In the following, we will use the notation $\cT_{[i]} := \cT_i \circ \cdots \circ \cT_2 \circ \cT_1$ and  $\tilde{\cT}_{[i]} := \tilde{\cT}_i \circ \cdots \circ \tilde{\cT}_2 \circ \tilde{\cT}_1$, for $1 < i \leq d-1$.

\smallskip In the Theorem~\ref{thm:cir-noise-ftec} below, we consider parallel error correction across multiple code blocks of the QLDPC code, as needed in the proof of Theorem~\ref{thm:main-cons-over}. We show that if the reduced weight of the input error is bounded in each code block, the single shot error correction suppresses the input error in each code block, even with general circuit noise applied to the syndrome extraction circuit. Consequently, errors can be suppressed in all the code blocks in parallel. It is worth emphasizing that the input error can be correlated across different code blocks; the only requirement of Theorem~\ref{thm:cir-noise-ftec} is that its weight with respect to each code block is bounded.

\begin{theorem} \label{thm:cir-noise-ftec}
    There exists constants $R, S, \beta > 0$, satisfying $\beta Sn + Rn < \frac{d_{\min}}{2}$, such that the following holds:
    
  Let $\delta$ be the error rate, satisfying $5\delta 2^{d+1} \leq S$, with $d$ being the depth of the syndrome extraction circuit $\Phi$. Let $\kappa := ed$ with $e := \exp(1)$, and $h$ be a positive integer, such that $h\kappa e \exp(-\frac{n\delta}{3}) \leq 1$. Consider parallel error correction of the $h$ code blocks of the QLDPC code $C$. For the $i^{th}$ code block, let $(D_i, A_i)$ denote the data and ancilla qubits, $\overline{L}^i$ denote the code space and $\Phi^i$ denote the syndrome extraction circuit. Let $\cV$ be a superoperator with the reduced weight $|\cV|_\mathrm{red} \leq \beta S n + Rn$ with respect to all the code blocks $i \in [h]$, and let $\ket{\psi} \in \otimes_i \overline{L}^i$ be any code state. Then, for any $\tilde{\cT}_{\Phi^i} \in \trans(\Phi^i, \delta), i \in [h]$,
\begin{equation} \label{eq:circ-noise-reduce}
\| \left(\otimes_{i \in [h]} (\tr_{A_i} \circ \: \cU^i_{ec} \circ \tilde{\cT}_{\Phi^i})\right) \left(\cV (\proj{\psi})\right) -  \cV' (\proj{\psi})  \|_1 < h\kappa e \exp(-\frac{n\delta}{3}),
\end{equation} 
where  $\cU^i_{ec}$ corresponds to the error correction channel on the $i^{th}$ block, $\tr_{A_i}$ refers to tracing out ancilla qubits in $A_i$, and $\cV' $ is a superoperator on $h$ code blocks, satisfying $|\cV'|_{\mathrm{red}} \leq \beta S n  + 5 n \delta 
2^{d+1}$, with respect to all the code blocks $i \in [h]$.
\end{theorem}
Note that if $\delta < \frac{R}{5 \times 2^{d+1}}$, then $|\cV'|_{\mathrm{red}} < |\cV|_{\mathrm{red}} $ in Theorem~\ref{thm:cir-noise-ftec}. Moreover, given that $h = \mathrm{poly}(n)$, the right hand side of Eq.~(\ref{eq:circ-noise-reduce}) approaches zero as $n$ goes to infinity. Therefore, the error weight can be effectively suppressed from $\beta S n + Rn$ to $\beta S n  + 5 n \delta 
2^{d+1}$, by choosing sufficiently large $n$. 

\smallskip Before proving Theorem~\ref{thm:cir-noise-ftec}, we provide Lemma~\ref{lem:pushing-errs}, which shows that for syndrome extraction circuit $\Phi$, any $\tilde{\cT}_\Phi \in \trans(\Phi, \delta)$ can be approximated by a superoperator $\overline{\cT}_\Phi  =   \cT_{\Phi_d} \circ \cN \circ \cT_{\Phi_{[d-1]}}$, where $\cN$ is a superoperator of weight $O(n\delta)$. Therefore, the errors that happen in the circuit $\Phi$ can be pushed just before the single qubit measurements are applied on the ancilla qubits. Intuitively, the weight of $\cN$ is $O(n\delta)$ because the constant depth of the circuit $\Phi$ limits the propagation of errors.
\begin{lemma}  \label{lem:pushing-errs}
For any $n$ such that $\exp(-\frac{n\delta}{3}) \leq \frac{1}{\kappa}, \kappa = ed$, and any $\tilde{\cT}_\Phi \in \trans(\Phi, \delta)$, there exists a superoperator $\overline{\cT}_\Phi$  such that 
    \begin{equation} \label{eq:T-bar-phi}
        \dnorm{\tilde{\cT}_\Phi - \overline{\cT}_\Phi} < \kappa \exp(-\frac{n\delta}{3}),
    \end{equation}
and,
\begin{equation} \label{eq:T-bar-phi-1}
   \overline{\cT}_\Phi  =   \cT_{d} \circ \cN \circ \cT_{[d-1]}, 
\end{equation}
where $\cN$ is a superoperator of weight $5n\delta 2^{d+1}$. 
\end{lemma}

\begin{proof}
We consider the first $d-1$ layers of  $\Phi$ (i.e., $\cT_{[d-1]} = \cT_{d-1} \circ \cdots \circ \cT_2 \circ \cT_1$) and the last layer (i.e., $\cT_{d}$), separately (see also Fig.~\ref{fig:syn-ext-circuit}). We approximate any $\tilde{\cT}_{[d-1]}$ by a superoperator $\overline{\cT}_{[d-1]} = \cN_{[d-1]} \circ \cT_{[d-1]}$, and $\tilde{\cT}_{d}$ by a superoperator $\overline{\cT}_{d} =  \cT_{d} \circ \cN_d$, where the weights of $\cN_{[d-1]}$ and $\cN_d$ are given by $5n\delta (2^{d + 1} - 4)$, and $10n\delta$ respectively. We then take $\overline{\cT}_{\Phi} :=  
\overline{\cT}_{d} \circ \overline{\cT}_{[d-1]}$, which satisfies Eq.~(\ref{eq:T-bar-phi}) and Eq.~(\ref{eq:T-bar-phi-1}).

\smallskip For each layer $i$ of the circuit, we do the following expansion,
    \begin{align}
        \tilde{\cT}_{i} &= \otimes_{j = 1}^{t_i} (\cT_{ij} + (\tilde{\cT}_{ij} - \cT_{ij}))  \nonumber \\
        &= \sum_{B \subseteq [t_i]} (\otimes_{j \in [t_i]\setminus B} \cT_{ij}) \otimes (\otimes_{j \in B} (\tilde{\cT}_{ij} - \cT_{ij})) \nonumber
    \end{align}
For $s > 0$, we define a superoperator $\overline{\cT}_{i}$ of weight $s$, as follows,
\begin{equation} \label{eq:T-i-bar}
    \overline{\cT}_{i} := \sum_{B \subseteq [t_i], |B| \leq s} (\otimes_{j \in [t_i]\setminus B } \cT_{ij}) \otimes (\otimes_{j \in B} (\tilde{\cT}_{ij} - \cT_{ij}))
\end{equation}
Note that
\begin{equation}
    \dnorm{\tilde{\cT}_i - \overline{\cT}_i} \leq \sum_{j > s}^{t_i} \binom{t_i}{j} \delta^j.
\end{equation}
Then, by doing a similar calculation as in Lemma~\ref{lem:bound-adv}, we have that for $s = 5 t_i \delta$ (using $t_i \geq n$)
\begin{equation}
    \dnorm{\tilde{\cT}_{i} - \overline{\cT}_{i}} \leq \exp(-\frac{t_i\delta}{3}) \leq \exp(-\frac{n\delta}{3}). 
\end{equation} 
For $\overline{\cT}_{[d - 1]} = \overline{\cT}_{d-1} \circ  \cdots \circ \overline{\cT}_{1}$,  we have 
\begin{align}
    \dnorm{ \tilde{\cT}_{[d-1]} -  \overline{\cT}_{[d-1]}} & \leq \dnorm{  \tilde{\cT}_{[d-1]} - \overline{\cT}_{d-1} \circ \tilde{\cT}_{[d-2]} } +  \dnorm{ \overline{\cT}_{d-1} \circ \tilde{\cT}_{[d-2]} -  \overline{\cT}_{[d-1]}} \nonumber\\
    &\leq \dnorm{\tilde{\cT}_{d-1} - \overline{\cT}_{d-1}} \dnorm{\tilde{\cT}_{[d-2]}} + \dnorm{\overline{\cT}_{d-1}} \dnorm{ \tilde{\cT}_{[d-2]} -  \overline{\cT}_{[d-2]}} \nonumber \\
    &\leq \exp(-\frac{n\delta}{3}) + \left(1 + \exp(-\frac{n\delta}{3})\right) \dnorm{ \tilde{\cT}_{[d-2]} -  \overline{\cT}_{[d-2]}}, \nonumber
\end{align}
where the first inequality uses the triangle inequality, the second inequality uses $\dnorm{\cT_2 \circ \cT_1} \break \leq \dnorm{\cT_2} \dnorm{\cT_1}$, and the last inequality uses $\dnorm{ \tilde{\cT}_{[d-2]}} = 1$ and             $\dnorm{\overline{\cT}_{d-1}} \leq \dnorm{\overline{\cT}_{d-1} - \tilde{\cT}_{d-1}} + \dnorm{\tilde{\cT}_{d-1}} \leq \left(\exp(-\frac{n\delta}{3}) + 1 \right)$. Applying the above recursively, we get
\begin{align} \label{eq:apprx-d-1}
     \dnorm{ \tilde{\cT}_{[d-1]} -  \overline{\cT}_{[d-1]}} &\leq (1 + (1 + \exp(-\frac{n\delta}{3})) + \cdots + (1 + \exp(-\frac{n\delta}{3}))^{d-2}) \exp(-\frac{n\delta}{3}) \nonumber\\
     & <   e(d-1)  \exp(-\frac{n\delta}{3}),
\end{align}
where in the second inequality, we have used $ \exp(-\frac{n\delta}{3}) \leq  \frac{1}{ed} < \frac{1}{d}$, which implies $\ln(1 + \exp(-\frac{n\delta}{3})) < \frac{1}{d}$, and in turn,  $(1 + \exp(-\frac{n\delta}{3}))^{d} < e $.

\smallskip Note from Eq.~(\ref{eq:T-i-bar}) that $ \overline{\cT}_{i}$ can be written as 
\begin{equation}
 \overline{\cT}_i = \sum_{B \subseteq [t_i],  |B| \leq s} (-1)^{|B| - |F|}\sum_{F \subseteq B} (\otimes_{j \in [t_i]\setminus F } \cT_{ij}) \otimes (\otimes_{j \in F} \tilde{\cT}_{ij}).  
\end{equation}
In other words, $\overline{\cT}_i $ can be written as a linear combination of the superoperators
\begin{equation}
    (\otimes_{j \in [t_i]\setminus F } \cT_{ij}) \otimes (\otimes_{j \in F} \tilde{\cT}_{ij}), \: |F| \leq s = 5 t_i \delta.
\end{equation} 
Since $\cT_{ij}, \forall i \in [d-1]$ corresponds to either a unitary or a state preparation (all the measurements are performed at the last layer), we have that
\begin{equation}
 \tilde{\cT}_{ij} = \cN_{ij} \circ \cT_{ij},
\end{equation}
where $\cN_{ij}$ is a quantum channel. Note that $\cN_{ij} = \tilde{\cT}_{ij} \circ \cT^\dagger_{ij}$ if $\cT_{ij}$ is a unitary and $\cN_{ij} =  \tilde{\cT}_{ij} \circ \tr$ if $\cT_{ij}$ is a state preparation channel, preparing a single qubit state. Let $F'$ be the set of qubits on which channels $\otimes_{j \in F} \cT_{ij}$ act.  Note that $|F'| \leq 2 |F|$ as $\Phi$ contains only single and two qubit gates. Let $\cN_{i, F'} := (\otimes_{j \in [t_i]\setminus F} \cI_{ij}) \otimes (\otimes_{j \in F} \cN_{ij})$, which has weight $|F'| \leq 2 |F| \leq 10t_i \delta$. We now have
\begin{align}
 (\otimes_{j \in [t_i]\setminus F } \cT_{ij}) \otimes (\otimes_{j \in F} \tilde{\cT}_{ij}) &=  \cN_{i,F'} \circ (\otimes_{j \in [t_i] } \cT_{ij}) \\
 &= \cN_{i,F'} \circ \cT_{i},
\end{align}
Therefore, we have $\overline{\cT}_i = \cN_i \circ \cT_i$, where $\cN_i$ is a superoperator of weight $10 t_i \delta$ (as it is linear combination of channels $\cN_{i, F'}$).

\medskip Consider the set of qubits $F'$ at the layer $i$. For $j > i$, let $\mathrm{Shade} (F', j)$ be the set of qubits after the $j^{th}$ layer of the quantum circuit $\Phi$ that depend on the qubits in $F'$. In other words, $\mathrm{Shade} (F', j)$ is the set of qubits in the future light cone of the set $F'$ after applying layers $i+1, \dots, j$. As $\Phi$ contains only single and two qubit gates, we have $\mathrm{Shade} (F', i+1) \leq 2 |F'|$. Therefore, 
\begin{equation} \label{eq:shade-cal}
 |\mathrm{Shade} (F', j)| \leq 2^{j-i} |F'| \leq 2^{j-i+1} |F| \leq 2^{j-i+1} 5t_i \delta < 5n\delta 2^{j-i+2},   
\end{equation}
where the second inequality uses $|F'| \leq 2 |F|$ and the last inequality uses $t_i \leq n + l < 2n$.

\smallskip Using the unitarity of $\cT_{i}, i \in [d-1]$, we can push $\cN_{i, F'}$ to the left as follows,
\begin{align} 
   \cT_{d-1} \circ \cdots \circ \cT_{i+1}  \circ (\cN_{i, F'} \circ \cT_{i})\circ \cdots  \circ \cT_{1} &=  \cT_{d-1} \circ \cdots \circ  \cN'_{i, \mathrm{Shade} (F', i + 1)} \circ \cT_{i+1}  \circ \cT_{i}\circ \cdots  \circ \cT_{1}  \nonumber\\ 
   &=  \cN'_{i, \mathrm{Shade} (F', d-1)} \circ \cT_{[d-1]}   \label{eq:push-err-i},
\end{align}
where in the first equality $\cN'_{i, \mathrm{Shade} (F', i + 1)} = \cT_{i+1} \circ \cN_{i, F'} \circ \cT_{i+1}^\dagger$ acts trivially except on $\mathrm{Shade} (F', i + 1)$. Similarly, in the second equality $\cN'_{i, \mathrm{Shade} (F', d-1)}$ is a quantum channel that acts non-trivially only on $\mathrm{Shade} (F', d-1)$, where $|\mathrm{Shade} (F', d-1)| \leq 2^{d-i+1} 5n \delta$ from Eq.~\eqref{eq:shade-cal}. We now apply Eq.~(\ref{eq:push-err-i}) for each layer $i = d-1, \dots, 1$ in that order and get the following,
\begin{equation} \label{eq:err-d-1}
    \overline{\cT}_{[d-1]} = \cN_{[d-1]} \circ \cT_{[d-1]}, 
\end{equation}
where $\cN_{[d-1]}$ is a linear combination of quantum channels that act non-trivially on at most $5 n \delta (2^2 + 2^3 + \cdots + 2^{d}) = 5n\delta (2^{d+1} - 4)$ qubits. Therefore, the weight of $\cN$ is given by $5n\delta (2^{d+1} - 4)$.

\medskip We now consider the $d^{th}$ layer of the circuit $\Phi$. We have $\cT_{d} = (\otimes_{j \in D} \cI_j ) \otimes (\otimes_{j \in A} \cM_j)$, where $\cM_j$ corresponds to a single qubit measurement in Pauli-$Z$ or $X$ basis (recall that $D$ is the set of data qubits with $|D| = n$ and $A$ is the set of ancilla qubits  with $|A| = l$). Consider now $\tilde{\cT}_{d}$ and the corresponding approximation $\overline{\cT}_d$ from Eq.~(\ref{eq:T-i-bar}), obtained by using $t_d = n + l$ ,
\begin{multline} \label{eq:approx-measure}
    \overline{\cT}_{d} = \sum_{ \substack{B_1 \subseteq D, B_2 \subseteq A \\ |B_1| + |B_2| \leq 5 (n+l) \delta }}  (\otimes_{j \in D \setminus B_1} \cI_{j}) \otimes (\otimes_{j \in A \setminus B_2} \cM_{j}) \\ \otimes (\otimes_{j \in B_1} (\tilde{\cI}_{j} - \cI_{j})) \otimes (\otimes_{j \in B_2} (\tilde{\cM}_{j} - \cM_{j})) \nonumber,
\end{multline}
which satisfies,
\begin{equation} \label{eq:apprx-d}
    \dnorm{\tilde{\cT}_{d} - \overline{\cT}_{d}} \leq \exp(-\frac{n\delta}{3}). 
\end{equation}
The superoperator $\overline{\cT}_{d}$ can be written as the linear combination of the following channels for a subset of qubits $F_1 \subseteq D, F_2 \subseteq A$, such that $|F_1| + |F_2| \leq 5(n+l)\delta < 10 n\delta$,  
\begin{equation}
    (\otimes_{j \in D \setminus F_1} \cI_{j}) \otimes (\otimes_{j \in A \setminus F_2} \cM_{j}) \otimes (\otimes_{j \in F_1} \tilde{\cI}_{j}) \otimes (\otimes_{j \in F_2} \tilde{\cM}_{j}).
\end{equation}
Consider single qubit measurement $\cM_j: \bL(\mathbb{C}^2) \to \mathbb{B}$ in the computational basis defined by $\cM_j(\cdot) = \sum_{a \in \{0, 1\}} \tr(\proj{a} \cdot) \proj{a}$. The noisy measurement $\tilde{\cM}_j : \bL(\mathbb{C}^2) \to \mathbb{B}$ corresponds to a quantum channel 
\begin{equation}
 \tilde{\cM}_j (\cdot) = \sum_{a \in \{0, 1\}} \tr(P_a \: \cdot) \proj{a},   
\end{equation}
where $P_a \geq 0$, $\sum_a P_a = \ident$. It can be seen that 
\begin{equation}
\tilde{\cM}_j = \cM_j \circ \cN_j,   
\end{equation}
where the channel $\cN_j$ is defined using qubit systems $A$ and $A'$ as below (see also Naimark's theorem~\cite[Theorem~2.42]{watrous2018theory})
\begin{equation}
  \cN_j = \tr_{A}(W_{A \to AA'} (\cdot) W^\dagger_{A \to AA'}),  \:  W_{A \to AA'} := \sum_{a \in \{0, 1\}}  \sqrt{P_a} \otimes \ket{a}_{A'}.
\end{equation}
Note that $W_{A \to AA'}$ is an isometry. In a similar way, for Pauli-$X$ basis measurement, we have $\tilde{\cM}_j = \cM_j \circ \cN_j$ for a quantum channel $\cN_j$. Therefore,
\begin{align}
   & (\otimes_{j \in D \setminus F_1} \cI_{j}) \otimes (\otimes_{j \in A \setminus F_2} \cM_{j}) \otimes (\otimes_{j \in F_1} \tilde{\cI}_{j}) \otimes (\otimes_{j \in F_2} \tilde{\cM}_{j})  \nonumber\\
   & =   (\otimes_{j \in D} \cI_i ) \otimes (\otimes_{j \in A} \cM_i) \circ \cN_{d, F_1, F_2} \nonumber \\
   &= \cT_{d} \circ \cN_{d, F_1, F_2},
\end{align}
where $\cN_{d, F_1, F_2} :=     (\otimes_{j \in D \setminus F_1} \cI_{j}) \otimes (\otimes_{j \in A \setminus F_2} \cI_{j}) \otimes (\otimes_{j \in F_1} \tilde{\cI}_{j}) \otimes (\otimes_{j \in F_2} \cN_{j}) $. Therefore, we have  
\begin{equation} \label{eq:err-d}
    \overline{\cT}_{d} = \cT_{d} \circ \cN_d,
\end{equation}
where $\cN_d$ is a superoperator of weight $10n\delta$. Finally, from Eq.~(\ref{eq:apprx-d-1}) and Eq.~(\ref{eq:apprx-d}), we get 
\begin{align}
\dnorm{\tilde{\cT}_\Phi - \overline{\cT}_\Phi} &\leq \dnorm{\overline{\cT}_{d}} \dnorm{\tilde{\cT}_{[d-1]} - \overline{\cT}_{[d-1]}} + \dnorm{\tilde{\cT}_{d} - \overline{\cT}_{d}} \nonumber \\
& \leq (1 + \exp(-\frac{n\delta}{3})) e (d-1) \exp(-\frac{n\delta}{3}) + \exp(-\frac{n\delta}{3})  \nonumber \\
& < ed \exp(-\frac{n\delta}{3}),
\end{align}
where the second line uses Eq.~(\ref{eq:apprx-d-1}) and Eq.~(\ref{eq:apprx-d})  and $\dnorm{\overline{\cT}_{d}} \leq  (1 + \exp(-\frac{n\delta}{3}))$, and the third line uses $1 + \exp(-\frac{n\delta}{3}) \leq 1 + \frac{1}{ed}  < 1 + \frac{1 - \frac{1}{e}}{d-1} = \frac{d - \frac{1}{e}}{d-1}$ (Recall that $d > 1$).

\smallskip Therefore, $\overline{\cT}_{\Phi} = \cT_{d} \circ \cN \circ \cT_{[d-1]}$, where $\cN := \cN_d \circ \cN_{[d-1]}$ is a quantum channel of weight $5n\delta 2^{d+1}$. Hence, Eq.~(\ref{eq:T-bar-phi}) and Eq.~(\ref{eq:T-bar-phi-1}) hold.
\end{proof}

\medskip We now prove Theorem~\ref{thm:cir-noise-ftec}. \newline

\subsubsection{Proof of Theorem~\ref{thm:cir-noise-ftec}}
Similar to Eq.~(\ref{eq:apprx-d-1}), using the triangle inequality, and the equality $\dnorm{\cT_1 \otimes\cT_2} = \dnorm{\cT_1} \dnorm{\cT_2}$ for superoperators $\cT_1, \cT_2$~\cite{kitaev}, we get from Eq.~(\ref{eq:T-bar-phi}) 
\begin{align}
    &\dnorm{(\otimes_{i \in [h]} \tilde{\cT}_{\Phi^i}) -  (\otimes_{i \in [h]} \overline{\cT}_{\Phi^i})} \nonumber\\
    & \leq \left(1 + (1 + \kappa \exp(-\frac{n\delta}{3})) + \dots + (1 + \kappa \exp(-\frac{n\delta}{3}))^{h-1} \right) \kappa \exp(-\frac{n\delta}{3}) \nonumber \\
    & \leq h (1 + \kappa \exp(-\frac{n\delta}{3}))^{h}  \: \kappa \exp(-\frac{n\delta}{3}) \nonumber \\
    & \leq h\kappa \exp(h \kappa \exp(-\frac{n\delta}{3})) \exp(-\frac{n\delta}{3}) \nonumber \\ 
    & < h \kappa e \exp(-\frac{n\delta}{3}), \label{eq:h-syndrome-blocks}
\end{align}
where in the second to last inequality, we have used $(1 +  x )^h \leq e^{xh}$, and the last inequality uses the assumption $h \kappa \exp(-\frac{n\delta}{3}) \leq \frac{1}{e} < 1$. From Eq.~(\ref{eq:h-syndrome-blocks}), we have 
\begin{multline} \label{eq:syndrome-meas}
\| \left(\otimes_{i \in [h]} (\tr_{A_i} \circ \: \cU^i_{ec} \circ \tilde{\cT}_{\Phi^i}) \right) \left(\cV (\proj{\psi} \right) - \left(\otimes_{i \in [h]} (\tr_{A_i}  \circ \: \cU_{ec}^i \circ \overline{\cT}_{\Phi^i}) \right) \left(\cV (\proj{\psi}) \right)  \|_1  \\  < h\kappa e \exp(-\frac{n\delta}{3}). 
\end{multline}

\smallskip We now show that for any superoperator $\cV$, such that $|\cV|_\mathrm{red} \leq \beta S n + Rn$ with respect to all the code blocks $i \in [h]$, we have 
\begin{equation}  \label{eq:fin-ftec-1}
\left(\otimes_{i \in [h]} (\tr_{A_i}  \circ \: \cU_{ec}^i \circ \overline{\cT}_{\Phi^i}) \right) \left(\cV (\proj{\psi}) \right) = \cV'(\proj{\psi}),   
\end{equation}
where $|\cV'|_{\mathrm{red}} \leq \beta S n + 5n\delta 2^{d+1}$ is satisfied for all the code blocks $i \in [h]$. Then, from Eq.~(\ref{eq:syndrome-meas}) and Eq.~(\ref{eq:fin-ftec-1}), we get Eq.~(\ref{eq:circ-noise-reduce}).
 
\medskip First, note that the first $d-1$ layers of the syndrome extraction circuit $\Phi$ realize an isometry on the data qubits
\begin{equation} \label{eq:syndrome-d-1-map}
    \cT_{[d-1]}(\cdot) = U_{[d-1]} (\cdot \otimes \proj{0}^{\otimes l}) U_{[d-1]}^\dagger,
\end{equation} 
where $U_{[d-1]}$ is a unitary such that for any $\ket{\phi} \in (\mathbb{C}^2)^{\otimes n}$, we have
\begin{equation} \label{eq:proj-syndrome}
    \left(\ident_D  \otimes \bra{\sigma}_A \right)  U_{[d-1]} (\ket{\phi} \otimes \ket{0}_A^{\otimes l}) =  \Pi_\sigma \ket{\phi},
\end{equation}
where $\Pi_{\sigma} \in \bL((\mathbb{C}^2)^{\otimes n})$ is the projector on the space corresponding to the syndrome $\sigma \in \{0, 1\}^l$. The last layer of $\Phi$ realizes single qubit Pauli-$Z$ measurements on ancilla qubits; therefore, the corresponding quantum channel is given by, 
\begin{equation} \label{eq:syndrome-d-map}
\cT_{d} = (\otimes_{j \in D} \cI_j) \otimes \sum_{\sigma \in \{0, 1\}^l} \proj{\sigma} (\cdot) \proj{\sigma}.
\end{equation}

We now focus on error correction for a fixed code block $j \in [h]$. Consider data qubits $D_j$ and the ancilla qubits $A_j$ corresponding to this block.  Recall from Lemma~\ref{lem:pushing-errs} that 
\begin{equation}
\overline{\cT}_{\Phi^j} = \cT^j_d \circ \cN^j \circ \cT^j_{[d-1]},    
\end{equation}
where the superoperator $\cN^j$ acts on the joint system $D_j \cup A_j$ and its weight is given by $ 5 n \delta  2^{d+1}$. Therefore,  $\cN^j$ can be written as a linear combination of superoperators of the form $W (\cdot) W'$,  where $W, W'$ are operators acting on the joint system $D_j \cup A_j$, and their weight is $5n\delta 
 2^{d+1}$ (see also Def.~\ref{def:weight}). Here, for simplicity, we replace $\cN^j$ by  $W(\cdot)W'$.  Similarly, we replace $\cV$ by $E (\cdot) E'$, where $E, E'$ are the operators that act on the joint code blocks $\cup_{i \in [h]} D_i$, satisfying $|E|_\mathrm{red}, |E'|_\mathrm{red} \leq \beta Sn + Rn$, for any code block $i \in [h]$.

\smallskip For an $\be_x \in \{0, 1\}^{n + l}$ indicating a Pauli-$X$ error on the data and ancilla qubits in the $j^{th}$ code block, we define $\be_x := (\be_x|_{D_j}, \be_x|_{A_j})$, where $\be_x|_{D_j} \in \{0, 1\}^{n}$ indicates the $X$ error on the data qubits and $\be_x|_{A_j} \in \{0, 1\}^l$ on the ancilla qubits. Similarly, we use $\be_z := (\be_z|_{D_j}, \be_z|_{A_j})$ for Pauli-$Z$ errors. For clarity, we use the following notation 
\begin{equation}
\ket{U^j_{[d-1]}E\psi} := \left((\otimes_{i \neq j} \ident_{D_i})   \otimes U^j_{[d-1]} \right)(E\ket{\psi} \otimes \ket{0}^{\otimes l}_{A_j}),    
\end{equation}
where $U^j_{[d-1]}$ is a copy of $U_{[d-1]}$ in Eq.~\eqref{eq:syndrome-d-1-map} and Eq.~\eqref{eq:proj-syndrome}. Note that the quantum state $\ket{U^j_{[d-1]}E\psi}$ corresponds to the joint data-ancilla state after the $(d-1)^{th}$ layer of syndrome extraction applied only to the $j^{th}$ code block. For any measurement outcome $\sigma \in \{0, 1\}^l$ corresponding to ancilla qubit measurements in the $d^{th}$ layer of the syndrome extraction circuit, we have 
\begin{align}
& \left((\otimes_{i \neq j} \ident_{D_i})  \otimes (\ident_{D_j} \otimes \bra{\sigma}_{A_j})  W  \right) \ket{U^j_{[d-1]}E\psi} \nonumber\\
&= \sum_{\be_x, \be_z \in \{0, 1\}^{n + l}} W_{\be_x, \be_z} \: \left((\otimes_{i \neq j} \ident_{D_i})  \otimes \left( (\ident_{D_j}  \otimes \bra{\sigma}_{A_j}) Z^{\be_z} X^{\be_x}  \right)\right) \ket{U^j_{[d-1]}E\psi}  \nonumber \\
& =\sum_{\be_x, \be_z}  W_{\be_x, \be_z} \: \left((\otimes_{i \neq j} \ident_{D_i})  \otimes (Z^{\be_z|_{D_j}} X^{\be_x|_{D_j}}  \otimes \bra{\sigma} Z^{\be_z|_{A_j}} X^{\be_x|_{A_j}} )\right) \ket{U^j_{[d-1]}E\psi} \nonumber \\
& = \sum_{\be_x, \be_z} (-1)^{\be_z|_{A_j} \cdot \sigma} W_{\be_x, \be_z} \: \left((\otimes_{i \neq j} \ident_{D_i})  \otimes (Z^{\be_z|_{D_j}} X^{\be_x|_{D_j}}  \otimes \bra{\sigma}  X^{\be_x|_{A_j}})\right) \ket{U^j_{[d-1]}E\psi} \nonumber \\
& = \sum_{\be_z|_{D_j}, \be_x}  W_{\be_z|_{D_j}, \be_x} \: \left((\otimes_{i \neq j} \ident_{D_i})  \otimes ( X^{\be_x|_{D_j}} Z^{\be_z|_{D_j}} \otimes \bra{\sigma \oplus \be_x|_{A_j}} )\right) \ket{U^j_{[d-1]}E\psi} \nonumber  \\
& = \sum_{\be_z|_{D_j}, \be_x}  W_{ \be_z|_{D_j}, \be_x} \: \left((\otimes_{i \neq j} \ident_{D_i})  \otimes X^{\be_x|_{D_j}} Z^{\be_z|_{D_j}} \Pi_{(\sigma \oplus \be_x|_{A_j})} \right) E\ket{\psi} \label{eq:proj-syndrome-1}
\end{align}
where in the first equality, we expand the operator $W$ in the Pauli basis, that is, $W = \sum_{\be_x, \be_z \in \{0, 1\}^{n+l}} W_{\be_x, \be_z} Z^{\be_z} X^{\be_x}$, and in the fourth equality, we have used $W_{\be_z|_{D_j}, \be_x} :=  \sum_{\be_z|_{A_j}} (-1)^{\be_z|_{A_j} \cdot \sigma}  W_{\be_x, \be_z}$. In the last equality, we use Eq.~(\ref{eq:proj-syndrome}). As the weight of $W$ is $5n\delta 2^{d+1}$, it follows that the Hamming weights of $\be_x|_{D_j},\be_z|_{D_j}, \be_x|_{A_j},\be_z|_{A_j}$ are upper bounded by $5n\delta 2^{d+1}$.

\smallskip Applying the error correction $U_{ec}$ on the data qubits corresponding to $j^{th}$ code block, when the observed syndrome is $\sigma \in \{0, 1\}^l$, we get 
\begin{align}
  &\left((\otimes_{i \neq j} \ident_{D_i}) \otimes  U_{ec}(\ident_{D_j} \otimes \proj{\sigma}_{A_j})  W  \right)  \ket{U^j_{[d-1]}E\psi} \nonumber \\
  & = \sum_{\be_z|_{D_j}, \be_x}   W_{\be_z|_{D_j}, \be_x} \: \left((\otimes_{i \neq j} \ident_{D_i})  \otimes U_{ec}(\ket{\sigma}_{A_j} \otimes  Z^{\be_z|_{D_j}} X^{\be_x|_{D_j}}   \Pi_{(\sigma \oplus \be_x|_{A_j})}) \right) E\ket{\psi}  \nonumber\\
  & = \sum_{\be_z|_{D_j}, \be_x}   W'_{\be_x, \be_z|_{D_j}} \: \left((\otimes_{i \neq j} \ident_{D_i}) \otimes Z^{\be_z|_{D_j}} X^{\be_x|_{D_j}}   U_{ec} (\ket{\sigma}_{A_j} \otimes \Pi_{(\sigma \oplus \be_x|_{A_j})}) \right) E\ket{\psi}  \nonumber\\
   & = \ket{\sigma}_{A_j} \otimes \left(\sum_{\be_z|_{D_j}, \be_x}  W'_{\be_x, \be_z|_{D_j}} \: \left((\otimes_{i \neq j} \ident_{D_i})  \otimes  Z^{\be_z|_{D_j}} X^{\be_x|_{D_j}} \right)  (E_{\neq j} \otimes  X^{\be'_x} Z^{\be'_z} )\ket{\psi} \right), \label{eq:ec-map}
\end{align}
where the first equality uses Eq.~(\ref{eq:proj-syndrome-1}), in the second equality, we use that $U_{ec}$ is a controlled Pauli, and $W'_{\be_x, \be_z|_{D_j}} = \pm \: W_{\be_x, \be_z|_{D_j}}$. In the last equality, from Corollary~\ref{cor:single-shot-dec}, $E_{\neq j} = \frac{1}{2^n} \sum_{\overline{\be}: \sigma_{\overline{\be}} = \sigma \oplus \be_x|_{A_j}} \tr_{D_j}(E((\otimes_{i \neq j} \ident_{D_i}) \otimes Z^{\overline{\be}_z} X^{\overline{\be}_x} ))$, and for $\be = (\be'_x, \be'_z)$ $|\be'|_{\mathrm{red}} \leq \beta S n$. Consider the total error operator in Eq.~(\ref{eq:ec-map})
\begin{equation}
    \overline{E} := \sum_{\be_x, \be_z|_{D_j}}  W'_{\be_x, \be_z|_{D_j}} \: \left((\otimes_{i \neq j} \ident_{D_i})  \otimes X^{\be_x|_{D_j}} Z^{\be_z|_{D_j}} \right)  (E_{\neq j} \otimes X^{e'_x} Z^{e'_z} )
\end{equation}
The stabilizer reduced weight of $\overline{E}$ with respect to the $j^{th}$ code block is upper bounded by $\beta Sn + 5 n \delta 2^{d+1}$. The reduced weight of $\overline{E}$ for code blocks $i \neq j$ remains upper bounded by $\beta Sn + Rn$. Similarly, after doing error correction with respect to all the code blocks, the residual error operator $\overline{E}$, satisfies $|\overline{E}|_{\mathrm{red}} \leq \beta Sn + 5 n \delta 2^{d+1}$ with respect to all the code blocks $i \in [h]$. Similarly, Eq.~(\ref{eq:ec-map}) also holds for the error operators $W'$ and $E'$; therefore, it follows that 
\begin{align} 
    & (\otimes_i (\tr_{A_i} \circ  \: \cU^i_{ec} \circ \overline{\cT}_{\Phi^i})) \circ \cV (\proj{\psi})  = \cV'(\proj{\psi}) := \sum_{\overline{E}, \overline{E}'} \overline{E} \proj{\psi} \overline{E}', \label{eq:fin-ftec}
\end{align}
where $|\cV'|_\mathrm{red} \leq \beta Sn +  5 n \delta 2^{d + 1}$,  with respect to all the code blocks. Finally, from Eq.~(\ref{eq:syndrome-meas}) and Eq.~(\ref{eq:fin-ftec}), we get Eq.~(\ref{eq:circ-noise-reduce}).  \hfill$\square$

\section*{Acknowledgements}
We would like to thank Paula Belzig, Alexander M\"uller-Hermes, Robert K\"onig, Thomas Theurer and Freek Witteveen for useful discussions.
MC and AG acknowledge financial support from the European Research Council (ERC Grant Agreement No.~818761), VILLUM FONDEN via the QMATH Centre of Excellence (Grant No.~10059) and the Novo Nordisk Foundation (grant NNF20OC0059939 ‘Quantum for Life’).  Part of this work was completed while MC was Turing Chair for Quantum Software, associated to the QuSoft research center in Amsterdam, acknowledging financial support by the Dutch National Growth Fund (NGF), as part of the Quantum Delta NL visitor programme.
OF acknowledges financial support from the European Research Council (ERC Grant, Agreement No.~851716) and from a government grant managed by the Agence Nationale de la Recherche under the Plan France 2030 with the reference ANR-22-PETQ-0006. We also thank the National Center for Competence in Research SwissMAP of the Swiss National Science Foundation and the Section of Mathematics at the University of Geneva for their hospitality.

\printbibliography
\end{document}